\documentclass[11pt]{article}

\usepackage{fullpage}
\usepackage{amsmath,algorithmic,algorithm,graphicx}
\usepackage{amssymb}
\usepackage{amsthm}
\usepackage{graphics}
\usepackage{color}
\usepackage{comment}

\usepackage{latexsym}
\usepackage{mathrsfs}
\usepackage{verbatim}

\ifx\pdftexversion\undefined
\usepackage[colorlinks,linkcolor=black,filecolor=black,citecolor=black,urlco
lor=black,pdfstartview=FitH]{hyperref}
\else
\usepackage[colorlinks,linkcolor=blue,filecolor=blue,citecolor=blue,urlcolor
=blue,pdfstartview=FitH]{hyperref}
\fi

\newcommand{\squishlist}{
 \begin{list}{$\bullet$}
  { \setlength{\itemsep}{0pt}
     \setlength{\parsep}{3pt}
     \setlength{\topsep}{3pt}
     \setlength{\partopsep}{0pt}
     \setlength{\leftmargin}{1.5em}
     \setlength{\labelwidth}{1em}
     \setlength{\labelsep}{0.5em} } }

\newcommand{\alert}[1]{\textbf{\color{red}
[[[#1]]]}\marginpar{\textbf{\color{red}**}}\typeout{ALERT:
\the\inputlineno: #1}}

\newcommand{\squishend}{
  \end{list}  }

\newtheorem{theorem}{Theorem}
\newtheorem{lemma}[theorem]{Lemma}
\newtheorem{definition}[theorem]{Definition}
\newtheorem{corollary}[theorem]{Corollary}
\newtheorem{proposition}[theorem]{Proposition}
\newtheorem{invariant}{Invariant}

\newtheorem{claim}[theorem]{Claim}

\newcommand{\E}{\mathbb{E}}
\newcommand{\sse}{\subseteq}

\newcommand{\diam}{{\rm diam}}

\newcommand{\N}[0]{{\ensuremath{\mathbb{N}}}}

\newcommand{\threat}{{\ensuremath{\mathcal{T}}}}
\newcommand{\threate}{{\ensuremath{\mathcal{J}}}}
\newcommand{\cE}{{\ensuremath{\mathcal{E}}}}
\newcommand{\cC}{{\ensuremath{\mathcal{C}}}}
\newcommand{\cR}{{\ensuremath{\mathcal{R}}}}
\newcommand{\cF}{{\ensuremath{\mathcal{F}}}}
\newcommand{\cFbar}{{\ensuremath{\overline{\mathcal{F}}}}}
\renewcommand{\S}{{\ensuremath{\mathcal{S}}}}
\newcommand{\SC}{{\ensuremath{\mathcal{S}_{|C}}}}

\newcommand{\cone}{\textrm{cone}}

\newcommand{\x}{\mathbf{x}}
\newcommand{\Texp}{\operatorname{Texp}}

\newcommand{\namedref}[2]{\hyperref[#2]{#1~\ref*{#2}}}
\newcommand{\sectionref}[1]{\namedref{Section}{#1}}

\newcommand{\theoremref}[1]{\namedref{Theorem}{#1}}
\newcommand{\propref}[1]{\namedref{Proposition}{#1}}
\newcommand{\corollaryref}[1]{\namedref{Corollary}{#1}}
\newcommand{\invref}[1]{\namedref{Invariant}{#1}}
\newcommand{\defref}[1]{\namedref{Definition}{#1}}

\newcommand{\claimref}[1]{\namedref{Claim}{#1}}
\newcommand{\lemmaref}[1]{\namedref{Lemma}{#1}}

\usepackage{pdfsync,color}

\newcommand{\agnote}[1]{}

\title{Cops, Robbers, and Threatening Skeletons: \\Padded Decomposition for Minor-Free Graphs\footnote{A preliminary version \cite{AGGNT14} of this paper appeared in STOC'14.}}
\author{
Ittai Abraham\thanks{VMWare. Email: \texttt{iabraham@vmware.com}.}
\and
Cyril Gavoille\thanks{LaBRI - University of Bordeaux, Bordeaux, France. Email: \texttt{gavoille@labri.fr}.}
\and
Anupam Gupta\thanks{Computer Science Department, Carnegie Mellon
  University, Pittsburgh, PA. Email: \texttt{anupamg@cs.cmu.edu}. Supported in part by
    NSF awards CCF-1016799 and CCF-1319811, and grant from the
    CMU-Microsoft Center for Computational Thinking.}
\and
Ofer Neiman\thanks{Department of Computer Science, Ben-Gurion University of the Negev, Beer-Sheva, Israel. Email: neimano@cs.bgu.ac.il. Supported in part by ISF grant No. (523/12) and BSF grant No. 2015813.}
\and
Kunal Talwar\thanks{Google. Email: \texttt{kunal@kunaltalwar.org}.}
}

\begin{document}

\maketitle

\begin{abstract}

We prove that any graph excluding $K_r$ as a minor can be partitioned into clusters of diameter at most $\Delta$ while removing at most $O(r/\Delta)$ fraction of the edges. This improves over the results of Fakcharoenphol and Talwar, who building on the work of Klein, Plotkin and Rao gave a partitioning that required to remove $O(r^2/\Delta)$  fraction of the edges.

Our result is obtained by a new approach that relates the topological
properties (excluding a minor) of a graph to its geometric properties
(the induced shortest path metric). Specifically, we show that
techniques used by Andreae in his investigation of the cops and
robbers game on graphs excluding a fixed minor, can be used to
construct padded decompositions of the metrics induced by such
graphs. In particular, we get probabilistic partitions with padding
parameter $O(r)$ and strong-diameter partitions with padding parameter
$O(r^2)$ for $K_r$-minor-free graphs, $O(k)$ for treewidth-$k$ graphs,
and $O(\log g)$ for graphs with (Euler) genus $g$.
\end{abstract}

\section{Introduction}

This paper considers the problem of constructing random partitioning
schemes for minor-free graphs. Loosely speaking, the goal is to find a
partition of the graph vertices so that each part (called a cluster)
has small diameter, and the probability of any local neighborhood
being cut (and not being contained within just one cluster) is
small. There is a natural tradeoff between these two parameters (the
diameter, and the probability of being cut).  Such random partitions
have found numerous applications in algorithm design, including:
flow/cut gaps, metric embeddings, and recently as core primitives for
several near linear time algorithms. Therefore improving the
parameters of the partitions is a research program of considerable
interest.

Tight parameters for such partitions are known in several settings. However, for the case of graphs that exclude some given graph $H$ as a minor, the problem of finding the optimal tradeoff remains open. Progress was made in the seminal work of Klein, Plotkin and Rao \cite{KPR93}, and improved by Fakcharoenphol and Talwar \cite{FT03}.
Despite attracting the attention of several researchers (see, e.g., \cite{JRL-blog}), the KPR framework remained
the only known approach to this problem for over 20 years.

In this paper we make progress on this question and improve known parameters.
Equally importantly, we also introduce techniques and structural insights that we
hope will be useful for further improvements on this and related
problems. In particular, we observe that the result of Andreae~\cite{Andreae86} can be reinterpreted as a \emph{structure theorem} for graphs excluding a fixed minor.
It constructively gives a {\em cop-decomposition} of a graph, which is a lot like a tree-decomposition except that instead of having $r$ vertices per bag, it guarantees having $r$ shortest-like paths in each bag.
The cop-decomposition gives weaker structure than the beautiful work of Robertson and Seymour \cite{RS03}, but has the benefit of significantly better dependence on $r$.
We extend this cop-decomposition framework to produce probabilistic partitions, and we believe that this high level approach may be useful in getting better algorithms for other problems involving excluded minor graphs.

We begin with some notation. For an undirected weighted graph $G=(V,E)$ and a
subset $C\subseteq V$, denote by $G[C]$ the induced subgraph on $C$. Let
$d_G$ denote the shortest path metric on $G$, and for $v\in V$ and $t\ge 0$ define the ball $B_G(v,t)=\{u\in V \mid d_G(v,u)\le t\}$. The \emph{(weak) diameter}
of a set $S \sse V$ is $\max_{x,y \in S} d_G(x,y)$, whereas the
\emph{strong-diameter} of the set $S$ is $\max_{x,y \in S}
d_{G[S]}(x,y)$ --- note that the latter distance is being measured in
the induced subgraph.

\begin{definition}[$\Delta$-bounded partitions]
  A partition $P=\{C_1,\dots,C_t\}$ of $V$ is \emph{$\Delta$-bounded} if for
  all $i$, the weak-diameter $\diam(C_i)\le \Delta$. The partition $P$ is {\em strong-diameter $\Delta$-bounded} if the strong-diameter $\diam(G[C_i])\le
  \Delta$ for all $i$.
\end{definition}

Given
a partition $P=\{C_1,\dots,C_t\}$ of $V$, let $P(z)$ denote the unique cluster
containing $z\in V$.

\begin{definition}
     A distribution $\mathcal{P}$ over $\Delta$-bounded partitions is \emph{$(\beta,\delta)$-padded} if for any $z\in
  V$ and any $0\le\gamma\le \delta $,
  \[
  \Pr[B_G(z, \gamma \Delta)\subseteq P(z)]\ge 2^{-\beta \gamma }~.
  \]
  We call $\mathcal{P}$ \emph{$\beta$-padded} if it is
  $(\beta,\delta)$-padded where $\delta$ is a universal constant that does not depend on $\beta$, and {\em efficient} if it can be sampled in polynomial time.
\end{definition}

Our definition of padded partitions is similar to the one in\cite{ABN11}, which generalizes several definitions that appeared before
, e.g. \cite{KPR93,GKL03,AGMW10}. In particular, our definition refers to cutting balls, and not only edges, and also allows for $\gamma>1/\beta$.

Our main result is the following.

\begin{theorem}\label{thm:main-weak}
  Every $K_r$-minor-free graph $G$ admits an efficient $O(r)$-padded partition scheme.
\end{theorem}

It has long been known that for arbitrary graphs the best possible
padding parameter is $\Theta(\log |V|)$ \cite{Bar96}. For special
cases better bounds are known, e.g., for metrics of doubling constant
$\lambda$, the padding parameter is $\Theta(\log\lambda)$
\cite{GKL03}. For graphs that can be drawn on a orientable surface of
genus $g$, ideas developed in a recent sequence of
papers~\cite{IndykS07,BorradaileLS10,Sidiropoulos10}
have culminated in the optimal padding parameter of
$\Theta(\log g)$~\cite{LeeS10}.

The first bounds for $K_r$-minor-free graphs were due to the influential
work of Klein, Plotkin, and Rao~\cite{KPR93}, who gave $(O(r^3),1/r)$-padded partition scheme. Fakcharoenphol and Talwar \cite{FT03} improved this to an
$(O(r^2),1/r)$-padded partition scheme. In this work, we improve
the padding parameter from $O(r^2)$ to $O(r)$; moreover, we provide
padding guarantees to larger balls --- the previous guarantees give
padding only for balls of diameter $< O(\Delta/r)$, compared to
$O(\Delta)$ for our result.
The partitioning scheme in~\cite{KPR93} was motivated by bounding
the maximum-multicommodity-flow/sparsest-cut gap for $K_r$-minor-free graphs.
Subsequently, it found applications to metric embeddings~\cite{Rao99,
  Yuri03} with its natural connections to edge-cut
problems~\cite{mat-book} and also to vertex-cut
problems~\cite{FeigeHL08}, to bounding higher eigenvalues and
higher-order Cheeger inequalities for graphs~\cite{BiswalLR10,
  KelnerLPT09, LeeGT12}, to metric extension problems and
approximation algorithms~\cite{CKR01-zero,AFHKTT,LN03}, and others.
The quantitative improvements given by our results thus give improvement in all these settings.

\theoremref{thm:main-weak} above gives us a weak-diameter guarantee.
However, our techniques are versatile, and can be extended to give
strong-diameter partitions --- in particular, we obtain the following
results.

\begin{theorem}\label{thm:stong-list}
Let $G$ be an undirected weighted graph.
\begin{enumerate}
\item If $G$ is a $K_r$-minor-free graph then it admits an efficient $(O(r^2),O(1/r^2))$-padded strong-diameter partition scheme.
\item If $G$ is a tree-width $r$ graph then it admits an efficient $(O(r),O(1/r))$-padded strong-diameter partition scheme.
\item If $G$ is a Euler-genus $g$ graph then it admits an efficient $O(\log g)$-padded strong-diameter partition scheme.
\end{enumerate}
\end{theorem}

The first result in \theoremref{thm:stong-list} is an exponential
improvement over the strong-diameter partitions of Abraham et. al.
\cite{AGMW10}. The third result strengthens the result of Lee and Sidiropoulos
\cite{LeeS10} by providing the same asymptotic padding guarantees
while ensuring that clusters have a strong-diameter. It holds for
graphs embedded on orientable or non-orientable surfaces of Euler
characteristic bounded by $2-g$ (see more details in
Section~\ref{sec:genus}). The second and the third results assume that
the embedding of the graph (into an optimal width tree-decomposition
or optimal Euler characteristic surface embedding) is given. Note that
such embeddings when $r$ or $g$ is bounded can be determined in
polynomial time (see for instance~\cite{Reed92,BK96,Mohar99,KMR08}).

\subsection{Discussion of Techniques}

How does one prove a property for a graph that does not contain a $K_r$ minor? One approach relies on the beautiful
results of Robertson and Seymour that turn this negative property, namely not having a certain minor, into a positive
constructive one. This gives a complete structural characterization of how such graphs are built from simple building
blocks by applying simple rules to them. This structure theorem allows one to prove properties of excluded minor graphs
by structural induction on the constructive procedure. On the negative side this approach typically inherits the rather
bad dependence on $r$ from the Robertson-Seymour structure theorem \cite{RS03}. Nevertheless, this approach has been highly successful and used to
prove several results for such graphs.

The other, somewhat more mysterious approach, is to work more directly and design an algorithm establishing the
property, such that by failing it constructs a $K_r$ minor. This approach is often problem-specific but usually leads
to better dependence on $r$. Examples of this approach include the work of Andreae \cite{Andreae86} for the Cops and Robbers game,
results of Alon, Seymour and Thomas \cite{AST90} on separators, and the aforementioned work of Klein, Plotkin and Rao \cite{KPR93}.

Let us now give a high-level description of some of the ideas and
techniques used to prove \theoremref{thm:main-weak}
and \theoremref{thm:stong-list}.

\paragraph{The Bounded Threatener Program and its probabilistic extension.}
A well-studied approach to obtain $\Delta$-bounded $\beta$-padded
probabilistic partitions is to find a set of ``suitable'' centers $S$,
and iteratively build balls around the points in $S$ with radii drawn
from a \emph{truncated exponential distribution} in the range $[\Delta/4,
\Delta/2]$ with rate $\beta$. The memoryless property of the exponential
distribution ensures that balls of radius $\approx \Delta/\beta$ around any vertex $z$ avoid
being cut with constant probability, conditioned on the exponential
distribution not being truncated. To handle the truncation, we need to
bound the number of centers at distance at most $(1/2+1/\beta)\Delta$
from any vertex $z$. We will call such centers the \emph{threateners}
of $z$. If the number of threateners is bounded by $2^{O(\beta)}$ then
a trivial union bound implies that with constant probability: none of them will reach diameter
$(1/2-1/\beta)\Delta$ and hence none of them will intersect the ball $B(z,\Delta/\beta)$.

A contribution of this work is in extending the bounded
threatener program and showing how a bound on the \emph{expected}
number of threateners suffices for obtaining probabilistic partitions.

\paragraph{Cop-Decompositions.} Andreae~\cite{Andreae86} considered the
following game, a set of cops plays against a robber. At each round the
robber can move across one edge and then each one of the cops can move
across one edge. The cops win if they land on the same vertex as the
robber. A key observation: if the robber is limited to a subgraph $V' \subset V$ and $P$ is a geodesic shortest path with respect to $G[V']$
then eventually a single cop can ``patroll'' $P$ and prevent the robber from even stepping on $P$.
Using this observation, Andreae showed that if $G$ is $K_r$-minor-free then $O(r^2)$ cops have
a winning strategy. The cop strategy is simple: each cop controls one
shortest path and collectively they try to iteratively build a $K_r$ minor. The shortest
paths controlled by the cops induce a set of supernodes (disjoint
connected subsets) and edges containing a minor that is a subgraph of
$K_r$. At each round one can fix a center for a new supernode and use
free cops to connect this center to all previous supernodes via shortest
paths. The new center and each new shortest path is fully contained in
the component containing the robber that is induced by removing the
supernodes from $G$ (hence these new paths are disjoint from all previous
supernodes).

We view Andreae's result as constructing a cop-decomposition of width $O(r)$, as we shall define now. First, recall that a tree-decomposition for $G$ is a tree $T$ whose nodes, called bags, are subsets of $V$ with the following properties: (1) $\cup_{B \in V(T)} B = V$; (2) for every edge $(u,v)$ of $G$, there is a bag of $T$ containing $u$ and $v$; and (3) for every $u\in V$, the set of bags containing $u$ induces a subtree of $T$.

\begin{definition}\label{def:cop}
A \emph{cop-decomposition} of \emph{width} $k$ for graph $G$ is a rooted tree $T$ that is a tree-decomposition for $G$ satisfying the following property. For every bag $B$ of $T$, the set of vertices of $B\setminus B'$, where $B'$ is the parent\footnote{If $B$ is the root bag of $T$, then we set $B'=\emptyset$.} bag of $B$, is composed of at most $k$ shortest paths of $G\setminus B'$.
\end{definition}

Note that the core difference between the width of a cop-decomposition and the width of a tree-decomposition is that we count the number of shortest paths instead of the number of vertices in each bag. The \emph{cop-width} of $G$, denoted by $cw(G)$, is the least number $k$ such that $G$ admits a cop-decomposition of width $k$. Observe that trees have cop-width $1$. If $G$ excludes a $K_r$ as a minor, then Andreae shows that $cw(G) \le r-1$. In fact, Andreae's cop algorithm constructively creates a cop-decomposition for $G$ of width $r-1$, moreover, each bag $B$ is actually a rooted shortest-path tree with at most $r-1$ leaves and whose root is in $B\setminus B'$.

\paragraph{From Cop-Decompositions to Padded Partitions via Skeletons.}
The cop-decomposition induces a partition of the vertices of the graph
into bags that consist of at most $r-1$ shortest paths. Note that the number of vertices in each bag in a
cop-decomposition may be large, and depend on $n$. Why are these bags
useful? Since each bag $B$ contains at
most $r-1$ shortest paths in the induced subgraph $B\setminus B'$ (where $B'$ is the parent of $B$), one can choose a
``net'' of centers along each path so that each node in the graph is
threatened by $O(r)$ centers from any one bag. Hence it now suffices to
bound the number of bags that get close enough to a vertex $z$ so that
some centers from this bag may threaten $z$. (We call such a bag a
\emph{``threatening skeleton''} for~$z$.) As mentioned above, we do not bound the
worst-case number of such threatening skeletons; we prove it suffices to bound
their expected number.

\paragraph{Bounding the Expected Number of Threateners.}
How to bound the expected number of threatening skeletons for some node
$z \in V$? We need a notion of progress. The cop-decomposition
ensures that in any given moment there are less than $r$ bags (a.k.a.\
threatening skeletons) that $z$
can see on the boundaries of its component, where each bag  consists of
a tree with at most $r-1$ paths. We observe the following property of the distances
from $z$ to these trees: if constructing a new tree $T_{new}$ in the induced
subgraph containing $z$ causes some current tree $T_{curr}$ to become farther
from $z$ (or even to be disconnected from $z$) because it cuts off some
short path from $z$ to $T_{curr}$, the distance from $z$ to
$T_{new}$ is strictly less than the distance from $z$ to $T_{curr}$. Indeed, if this distance were to miraculously decrease (deterministically)
by $\Delta/k$ then one can prove a bound of $O{r+k \choose k}$ on the
number of threateners. But why should such a large decrease happen? It
doesn't, but  \emph{we force this to happen in expectation}. We change the above
construction and build a ``buffer'' of some random radius around each
skeleton we build. Note that the supernodes did not have to be trees in
the above arguments,
and hence ``fattening'' them by growing buffers around the trees would not change any of the
preceding arguments.  Now by choosing the buffer radius from a truncated
exponential with rate $O(r)$, we may na\"{\i}vely hope to decrease the
distance by $\Delta/r$ with constant probability (assuming no
truncation). The proof is much more subtle, and requires to overcome the
truncation of the buffer. We use a potential function with delicately
chosen parameters,
such
that for each new tree, this potential increases in expectation by
$\approx r/2^r$. The potential starts at $0$ and once it reaches $r$, it means that $z$ is at
distance $0$ from some buffered tree and will not be threatened
again. Finally, the optional stopping theorem helps us bound the
expected number of threateners by $\approx 2^r$.

\paragraph{Bounding Expected Increase in Potential.}
In order to bound the number of threateners for $z$, the potential
function we use is a sum of exponentials $\sum_{\text{buffers} B} e^{
  - \alpha\, d(z,B)}$ for some parameter $\alpha$; the sum is over those
buffered trees that the node $z$ can see. The main
challenge is that in the worst case, one new buffered tree can cause
all the other current buffered trees to be disconnected from the
component containing $z$, hence losing $r$ summands of the potential. To
overcome this we need to guarantee that the expected gain from the new
tree is $O(r)$ times more than the expected loss of any single current
tree, which is one of the technical cores of the analysis. %
We note that obtaining any deterministic
bound on the number of threateners using a cop-decomposition, rather than
only bounding the expectation, remains an open question.

\subsection{Other Related Work}
\label{sec:related}

The ideas of either finding a ``good'' decomposition or else building a
$K_r$-minor used by~\cite{KPR93, Andreae86} also appear in
``shallow-minor theorems'' of Alon, Seymour, and Thomas~\cite{AST90},
Plotkin, Rao, and Smith~\cite{PRS94}, and others. The parameters and
run-times of these constructions have been considerably improved, see the paper of
Wulff-Nilsen~\cite{W11} and the references therein.

Busch, LaFortune, and Tirthapura~\cite{BLT07} first suggested the idea of
decomposing a graph into paths and building balls around these paths;
they considered this in the context of strong-diameter covers. They give the best constants for covers of
planar graphs; for $K_r$-minor-free graphs, they give $O(1)$-padding and
$O(\log |V|\cdot f(r))$-overlap, where $f(r)$ depends on the Robertson-Seymour
structure theorem.

In contrast to the weak-diameter partitions of Klein et. al. and  Fakcharoenphol and Talwar \cite{KPR93, FT03}, the previously
best \emph{strong-diameter} partitions are due to Abraham et. al. \cite{AGMW10},
who guarantee strong-diameter $\Delta$ and probability of an edge
$\{u,v\}$ being separated is $O(6^r r^2 \cdot \frac{d(u,v)}{\Delta})$. Abraham et. al. \cite{AGMW10} also
present \emph{sparse covers} with strong-diameter $\Delta$, padding of $O(r^2)$ and overlap of $2^{O(r)} r!$.

The papers~\cite{IndykS07,BorradaileLS10,Sidiropoulos10} give algorithms
to probabilistically embed genus-$g$ graphs into planar graphs with
$2^{O(g)}$, $O(g^2)$ and $O(\log g)$ distortion respectively. The ideas
developed in this line of work lead to an asymptotically optimal padding
parameter of $O(\log g)$ for genus-$g$ graphs~\cite{LeeS10}.

For general graphs, the decomposition schemes in,
e.g.,~\cite{A85,LS93,Bar96,CKR01-zero,FRT03} give asymptotically optimal
$O(\log |V|)$ padding. The best result known for tree-width-$r$ graphs was
the same as for $K_r$-minor-free graphs, i.e., $O(r^2)$-padding partitions.

\subsection{Organization of the Paper}

After a few preliminary definitions, we provide in \sectionref{sec:tools} a bound on the expected number of threateners for a wide range of partition algorithms, and show how to use this to bound the padding probability.
Our main result \theoremref{thm:main-weak} is proved in \sectionref{sec:weak-diam}. The three assertions of \theoremref{thm:stong-list} are then proven in Sections~\ref{sec:strong-diam}, \ref{sec:tree} and \ref{sec:genus}.

\section{Definitions and Notation}\label{sec:pre}

\paragraph{Graphs.}
We assume familiarity with graph-theoretic notions; see,
e.g.,~\cite{diestel} for background. Here are some definitions we will
use. Given a graph $G=(V,E)$, a ball around $A\subseteq V$ of radius $t\ge 0$ is $B_G(A,t)=\{u\in V \mid d_G(A,u)\le t\}$. Also let $N(A)=\{u\in V \mid \exists v\in A,~\{u,v\}\in E\}$. For subsets $A,B\subseteq V$ define a relation $\sim$ where $A\sim B$ iff $A\cap N(B)\neq\emptyset$, that is, iff there is an edge between a vertex of $A$ to a vertex of $B$.

A \emph{minor} of $G$ is a subgraph of a graph obtained from $G$ by a sequence of edge contractions.
Equivalently, $G'$ is a minor of $G$ if
there exists a map $f: V(G) \to V(G')$ such that (a)~ for each $u' \in
V(G')$ the ``supernode'' $f^{-1}(u')$ is connected in $G$, and (b)~for every edge $\{u',v'\} \in E(G')$, there is at least one edge between
$f^{-1}(u')$ and $f^{-1}(v')$ in $E(G)$.
A graph $G$ is
\emph{$H$-minor-free} (or \emph{excludes an $H$-minor}) if $G$ does not
contain a minor isomorphic to $H$. As is well-known,
planar graphs are exactly the graphs excluding $K_{3,3}$ and $K_5$ as
minors. In fact, Robertson and Seymour proved that every graph family
closed under taking minors is characterized by a finite set of excluded
minors.

Many one-way implications are also known: if we can show that a class
$\mathscr{G}$ of graphs is closed under taking minors, and $H \not\in
\mathscr{G}$, then $\mathscr{G}$ contains only $H$-minor-free-graphs. Hence, graphs with
tree-width at most $r$ are $K_{r+2}$-minor-free (since tree-width of a
clique is one smaller than its size, and the tree-width of a graph does
not increase under taking minors); graphs with genus
$g$ exclude $K_r$ as a minor for some $r = \Theta(\sqrt{g})$, since the genus
of $K_r$ is $\Theta(r^2)$.

\paragraph{Truncated Exponential Distributions.} We will extensively use
the following probability distribution over positive reals.
The \emph{$[\theta_1,\theta_2]$-truncated exponential distribution} with
parameter $b$ is denoted by $\Texp_{[\theta_1, \theta_2]}(b)$, and
the density function is:
\begin{equation}\label{eq:texpgen}
  f_{texp; b; \theta_1, \theta_2}(y) := \frac{ b\, e^{-b\cdot y} }{e^{-b \cdot \theta_1} - e^{-b
      \cdot \theta_2}} \qquad\qquad
  \text{ for } y \in [\theta_1, \theta_2].
\end{equation}
For the \emph{$[0,1]$-truncated exponential distribution} we drop the
subscripts and denote it by $\Texp(b)$; the density function is
\begin{equation}\label{eq:texpunit}
f_{texp; b}(y) := \frac{ b\, e^{-b\cdot y} }{1 - e^{-b}} \qquad\qquad
\text{ for } y \in [0,1].
\end{equation}

Note that if $Y\propto\Texp(b)$ then $u\cdot Y\propto\Texp_{[0,u]}(b/u)$.

\section{Analysis}\label{sec:tools}

Our algorithms induce an iterative process that creates ``skeletons''
(e.g., trees, paths, or vertices) and remove their neighborhoods (a
buffer), defined according to some truncated exponential distribution,
from the graph.  Once we have these skeletons, our algorithms define a
second iterative process that creates clusters from the
skeletons. %

Let us abstract out the properties needed from our first and second processes.

\begin{definition}[Skeleton-Process]\label{def:pro}
Given a graph $G$, parameters $0\le l<u\le 1$ and $b>0$, any process which generates a sequence of graphs $G=G_0,G_1,\dots$, skeletons $A_0,A_1,\dots$ and vertex sets $K_0,K_1,\dots$, that satisfies the following property is a {\em skeleton-process}:
\begin{itemize}
\item For any $i\ge 0$, $A_i\subseteq V(G_i)$ and $K_i=B_{G_i}(A_i, R_i\Delta )$, where $R_i \propto \Texp_{[l,u]}(b/(u-l))$.
\end{itemize}
The process is {\em threatening} if the graph sequence satisfies
$G_{i+1}=G_i\setminus K_i$, and  the process is {\em cutting} if the graph sequence satisfies $G_{i+1}\supseteq G_0\setminus(\cup_{j\le i}K_j)$.
\end{definition}

The first process is a threatening process which creates buffers around the trees of the cop-decomposition. The second process is a cutting process that creates the actual clusters centered at net-points of the trees. For the strong-diameter results, we will have a single process that satisfies both definitions.

\subsection{Analysis of the Threatening Process: Bounding the Expected Threats}\label{sec:threat}

A crucial property of all of our algorithms is that any vertex $z$ can ``see'' at most $s$ buffers (the $K_i$ sets) at any time, for some parameter $s$ (in the weak-diameter partition we will have $s=r$). By this we mean that for any connected component $C$ in one of the remaining graphs (after some buffers were removed), there are at most $s$ buffers that are connected to $C$ by an edge of $G$. This property will enable us to prove that $z$ is expected to be ``threatened'' by a small number of skeletons, that is, we expect a few skeletons that are sufficiently close to cut a certain ball around $z$.

Consider a threatening skeleton-process with parameters $l=0$, $u\in(0,1]$ and $b=2s$. We prove a bound on the expected number of threateners for a ball around any vertex $z$ of $G$ with padding parameter $\gamma>0$. For some $u\le u'\le 1$, let $\threate_z=\{A_i \mid  d_{G_i}(z,A_i)\le (u'+\gamma)\Delta\}$ be the set of vertex sets whose subset $K_i$ may intersect $B_z=B_G(z,\gamma\Delta)$. Observe that
once $z\in K_j$ for some index $j$ then it is removed from the graph,
and $\threate_z$ cannot increase anymore. For a connected component
$C_i\in G_i$ let ${\cal K}_{|C_i}=\{K_j \mid j<i \wedge C_i\sim
K_j\}$. (Recall that $A \sim B$ if there exists an edge from a node in
$A$ to some node in $B$.)
\begin{lemma}\label{lem:threat}
Suppose that in a threatening skeleton-process we have the property that for every $i\in\N$ and every connected component $C_i\in G_i$, we are guaranteed that $|{\cal K}_{|C_i}|\le s$, then
  \[
  \E[|\threate_z|]\le 6e^{(2s+1)\cdot(u'+\gamma)/u}~.
  \]
\end{lemma}
\begin{proof}
Fix any $i\in\N$. W.l.o.g., we may assume that the process always picks
the set $A_i$ in the connected component $C_i$ of $G_i$ that contains
$z$ (the other components do not affect ${\cal J}_z$). Let $\x=\x(i)$ be a vector of the ``normalized distances''
from $z$ to ${\cal K}_{|C_i}$. More formally, if
${\cal K}_{|C_i}=\{K_{i_1},\dots K_{i_l}\}$ (with $l\le s$ by the assumption of the lemma), then for $j\in[l]$ define
$$x_j := \frac{d_{G_i\cup K_{i_j}}(z,A_{i_j})-R_{i_j}\Delta}{u\Delta}~.$$ Intuitively, $x_j$ should have been the distance from $z$ to $K_{i_j}$, normalized by $u\Delta$. Note that by the definition of $K_{i_j}$ we have that $d_{G_i\cup K_{i_j}}(z,K_{i_j})\ge d_{G_i\cup K_{i_j}}(z,A_{i_j})-R_{i_j}\Delta$.

Define the potential function for the vector $\x:=(x_1,\dots,x_l)$ as
\begin{equation}\label{eq:pot}
\Phi(\x)=\sum_{j=1}^le^{-(2s+1)\cdot x_j}~. 
\end{equation}

We would like to analyze the change to $\x$ over time. Assume w.l.o.g that $x_1\le\dots\le x_l$. Let $h := \frac{d_{G_i}(z,A_i)}{u\Delta}\ge 0$ be the normalized distance of $z$ from
the set $A_i$, and let $y=h-R_i/u$.
Observe that if $x_j \leq y$ then the
shortest path from $K_{i_j}$ to $z$ is completely disjoint from $K_i$;
seeking contradiction, assume $a\in K_i$ lies on the shortest path in $G_i$ from $z$ to $K_{i_j}$.
Since every vertex of distance $R_{i_j}\Delta$ from $A_{i_j}$ is in $K_{i_j}$ and thus was removed from the graph, it must be that $d_{G_i\cup K_{i_j}}(a,A_{i_j})>R_{i_j}\Delta$. We conclude that
\[
d_{G_i}(z,a)<d_{G_i\cup K_{i_j}}(z,A_{i_j})-R_{i_j}\Delta=x_j\cdot(u\Delta)\le y\cdot(u\Delta)=d_{G_i}(z,A_i)-R_i\Delta\le d_{G_i}(z,K_i)\le d_{G_i}(z,a)~,
\]
contradiction.

We get that if $j^*$ is the
maximal index such that $x_{j^*}\le y$, then the first $j^*$ entries of $\x$
will not change. The new set $K_i$ will always be in ${\cal
  K}_{|C_{i+1}}$ (recall that $C_{i+1}$ is the component containing $z$
in $G_{i+1}=G_i\setminus K_i$), so we have that the $j^*+1$ entry in
$\x(i+1)$ will be $x_{j^*+1}=y$. For $j^*<j\le l$, it could be the case
that $K_i$ intersects the shortest path from $K_{i_j}$ to $z$, in which case
the distance may increase or $K_{i_j}$ can even be disconnected from
$z$. Note that if $l=s$, then it must be that at least one $K_{i_j}$ is
disconnected from $z$, because we assume that $|{\cal K}_{|C_{i+1}}|\le
s$.

Next we attempt to bound the expected change to the potential function $\Phi$ in any single step. To this end, it suffices to consider the worst scenario, in which all the $K_{i_j}$ for $j^*<j\le l$ become disconnected from $z$ by $K_i$ (in such a case the potential decreases the most). Define the ``filtered subsequence'' $\x \downarrow y$ to be the sequence obtained by dropping all the coordinates of $\x$ which are strictly
larger than $y$, and adding in $y$. (E.g., $(-0.4, -0.3, 0.7, 5, 6.9)
\downarrow 1.42 = (-0.4, -0.3, 0.7, 1.42)$.) So we assume that $\x(i)=\x(i-1)\downarrow y$ (where $y$ is define as above).
Define
\[
\Phi_i=\left\{\begin{array}{ccc} \Phi(\x(i)) & \forall j,~\forall i'
                                               \leq i, \x(i')_j> 0\\
\Phi_{i-1}+ 2s& \text{otherwise} \end{array} \right.
\]

Fix any sequence $\x(0),\dots,\x(i-1)$ (each of length at most $s$), which determines $\Phi_{i-1}$, and fix any $h\ge 0$. Recall that $y=h-Y$ with $Y \propto \Texp(2s)$ as in \eqref{eq:texpunit}. %

\begin{claim}
  \label{claim:phi}
    $\E_Y[\Phi_i-\Phi_{i-1}]\ge (s/2)\cdot e^{-(2s+1)h}$.  
\end{claim}

\begin{proof}
If it is the case that some $\x(i')$ for $i'< i$ has nonpositive coordinate, then
\[
\Phi_i=\Phi_{i-1}+ 2s\ge \Phi_{i-1}+\frac{s\cdot e^{-(2s+1)h}}{1-e^{-2s}}~,
\]
using that $e^{-2s}\le 1/2$ and $e^{-(2s+1)h}\le 1$ (since $h\ge 0$).
So from now on we assume $\x(i')$ has all positive coordinates for all
$i'< i$. Observe that $\x(i)$ will have nonpositive coordinate iff
$Y\ge h$, so we consider the two cases separately. Denote
$\bar{h}=\min\{1,h\}$. The first case is $Y\ge \bar{h}$, so we have
exactly $2s$ increase in potential. The second case is that
$Y<\bar{h}$, in which case we have $\Phi_i=\Phi(\x(i))$. Conditioning
on the event $\{Y<\bar{h}\}$ means that we sample $Y$ from the
distribution $\Texp_{[0, \bar{h}]}(2s)$ with density function
$f_{texp;2s;0,\bar{h}}$. (Indeed, $Y$ is 
$\Texp_{[0, 1]}(2s)$ random variable, and truncating a smaller value
    $\bar{h} \leq 1$ means we sample from the distribution $\Texp_{[0, \bar{h}]}(2s)$.)

The increase of the potential due to the new coordinate $y=h-Y$ is
  $e^{-(2s+1)\cdot (h-Y)}$, so the expected gain is
  \begin{eqnarray*}
    \E[ e^{-(2s+1)\cdot(h-Y)}\mid Y<\bar{h}] &=& e^{-(2s+1)h}\cdot\int_{0}^{\bar{h}} e^{(2s+1)w}\, f_{texp;2s;0,\bar{h}}(w)\, dw\\
     &=&    e^{-(2s+1)h}\cdot\int_{0}^{\bar{h}} \frac{2s\cdot e^w}{1 - e^{-2s\bar{h}}}\, dw \\
     &=&    \frac{2s(e^{\bar{h}}-1)\cdot e^{-(2s+1)h}}{1 - e^{-2s\bar{h}}}~,
  \end{eqnarray*}
Next we analyze the loss in potential for the coordinates $x_j$ that are
  dropped. Recall that a coordinate $x_j$ is dropped exactly when
  $x_j>h-Y$. Since we condition on $Y\in[0,\bar{h}]$, the only interesting case is when
  $x_j=h-\gamma$ for some $\gamma\in [0,\bar{h}]$, which is dropped when
  $Y>\gamma$. As $\x(i-1)$ has at most $s$ coordinates, the expected loss, conditioned on the event $\{Y<\bar{h}\}$, is at most
  \begin{eqnarray*}
    s\cdot \max_{\gamma \in [0,\bar{h}]} \left\{e^{-(2s+1)\cdot(h-\gamma)} \Pr[ Y > \gamma\mid Y<\bar{h} ]\right\} &=&
    s\cdot e^{-(2s+1)h}\cdot\max_{\gamma \in [0,\bar{h}]} \left\{e^{(2s+1)\gamma}\int_{\gamma}^{\bar{h}}
    f_{texp;2s;0,\bar{h}}(w)\, dw\right\}\\
    &=& s\cdot e^{-(2s+1)h}\cdot\max_{\gamma \in [0,\bar{h}]} \left\{e^{(2s+1)\gamma}\cdot \frac{e^{-2s\gamma} -
      e^{-2s\bar{h}}}{1 - e^{-2s\bar{h}}}\right\}\\
    &=&s\cdot \frac{e^{-(2s+1)h}}{1-e^{-2s\bar{h}}}\cdot\max_{\gamma\in [0,\bar{h}]}\left\{e^\gamma-e^{(2s+1)\gamma-2s\bar{h}}\right\}\\
    &\le&  s\cdot \frac{e^{-(2s+1)h}}{1-e^{-2s\bar{h}}}\cdot e^{\bar{h}}~.
  \end{eqnarray*}
We conclude that
\[
\E_Y[\Phi_i-\Phi_{i-1}\mid Y<h]\ge \frac{2s(e^{\bar{h}}-1)\cdot e^{-(2s+1)h}}{1 - e^{-2s\bar{h}}}-s\cdot \frac{e^{-(2s+1)h}}{1-e^{-2s\bar{h}}}\cdot e^{\bar{h}}=s(e^{\bar{h}}-2)\cdot \frac{e^{-(2s+1)h}}{1-e^{-2s\bar{h}}}
\]
Thus the expected increase in potential is
\begin{eqnarray*}
\E[\Phi_i-\Phi_{i-1}] &=&  \E[\Phi_i-\Phi_{i-1}\mid Y<\bar{h}]\cdot \Pr[Y<\bar{h}]+\E[\Phi_i-\Phi_{i-1}\mid Y\ge \bar{h}]\cdot \Pr[Y\ge \bar{h}]\\
&\ge&s(e^{\bar{h}}-2) \cdot  \frac{e^{-(2s+1)h}}{1 - e^{-2s\bar{h}}}\cdot \frac{1-e^{-2s\bar{h}}}{1-e^{-2s}}+2s\cdot\frac{e^{-2s\bar{h}}-e^{-2s}}{1-e^{-2s}}\\
  &=& \frac{(e^{\bar{h}}-2) + 2 (e^{-2s\bar{h}} - e^{-2s})
      e^{(2s+1)h}}{1 - e^{-2s}} \cdot s\, e^{-(2s+1)h}~.
\end{eqnarray*}
It now suffices to show that the first term is at least $\frac12$;
since the denominator is at most $1$, we focus on the
numerator. If $h \geq 1$ then $\bar{h} = \min(h,1) = 1$, we get
$e - 2 > 0.5$. Else $h \in [0,1]$ and so $\bar{h} = h$, which simplifies the expression
to
\begin{gather}
  (e^h-2) + 2 (e^{-2sh} - e^{-2s}) e^{(2s+1)h} = (e^h - 2) + 2e^h(1 -
  e^{-2s(1-h)}).
\end{gather}
This an increasing function of $s$, so smallest when $s = 1$. Now the
resulting expression $(e^h-2) + 2 (e^{-2h} - e^{-2}) e^{3h}$ is
unimodal for $h \in [0,1]$, and minimized at $h = 1$, again
giving $e-2 > 0.5$. This completes the proof. 
\end{proof}

Define $\zeta:=(s/2)\cdot e^{-(2s+1)\cdot(u'+\gamma)/u}$.%
\agnote{Changed it so that as soon as we get a negative coordinate,
  the process moves to adding $\zeta$ and stays there forever.}
Recall that for every $A_i\in\threate_z$ we have that $d_{G_i}(z,A_i)\le(u'+\gamma)\Delta$ and thus $h=d_{G_i}(z,A_i)/(u\Delta)\le (u'+\gamma)/u$. Observe that the expectation of Claim~\ref{claim:phi} is taken only over the current choice of $Y$, and since $Y$ is chosen independently we can condition on any other event that depends on previous steps, and obtain the same bound. In particular, for $A_i\in\threate_z$,
\begin{equation}\label{eq:gg}
\E[\Phi_{i+1}-\Phi_i\mid \Phi_i]\ge  \zeta~.
\end{equation}
Also note that the bound of Claim~\ref{claim:phi} is always positive, so even if $A_i\notin\threate_z$ we still have
\begin{equation}\label{eq:ggg}
\E[\Phi_{i+1}-\Phi_i\mid \Phi_i]\ge 0 ~.
\end{equation}
For $t\in \N$ let $j_t=|\{i~:~A_i\in \threate_z \mbox{ and } i\le
t\}|$ be the number of time steps until $t$ in which $z$ is
threatened.

Let $\cR_t$ be the $\sigma$-field defined by the independent $R_i$
variables observed until the $t^{th}$ step of this process, so that
$\{\cR_t\}_t$ forms a filtration. \agnote{Reworded}
We claim that the process $X_0,X_1,\dots$ where $X_t=\Phi_t-\zeta\cdot
j_t$, is a submartingale adapted to this filtration. To prove this consider two cases: If $A_{t+1}\in \threate_z$ then $j_{t+1} = j_t+1$, and by \eqref{eq:gg} we get $\E[\Phi_{t+1}\mid \Phi_t,j_t]\ge\Phi_t+\zeta$, and so
\[
\E[X_{t+1}\mid \cR_t]=\E[\Phi_{t+1}-\zeta\cdot j_{t+1}\mid \Phi_t,j_t]\ge \Phi_t+\zeta-\zeta\cdot j_{t+1} = \Phi_t+\zeta-\zeta\cdot (j_t+1) = X_t~.
\]
If it is the case that $A_{t+1}\notin \threate_z$, then $j_{t+1}=j_t$ and  by \eqref{eq:ggg} 
\[
\E[\Phi_{t+1}-\zeta\cdot j_{t+1}\mid \Phi_t,j_t]\ge \Phi_t-\zeta\cdot j_t=X_t~.
\]

\agnote{Reworded}
The stopping time of this (sub)martingale $X_0,X_1,\dots$ is a random
variable $\tau$ that has support in $\N$, and such that the event
$\tau=t$ is measurable with respect to the filtration $\cR_t$. %
 Define $\tau$ as the first time in which $\x(\tau)$ has a nonpositive coordinate. Observe that if $t$ is the time where $z\in K_t$, then it must be that $d_{G_t}(z,A_t)\le R_t\Delta$, and so we get a nonpositive coordinate in $\x(t)$ which implies that $\tau=t$. Since the stopping time is bounded by $|V|$ (there can be at most $|V|$ rounds, because at least one vertex is removed every round), we can apply Doob's optional stopping time Theorem \cite[Section~12.5]{GS01} and obtain that
\[
\E[\Phi_\tau]-\zeta\cdot\E[j_\tau]=\E[X_\tau]\ge\E[X_0]=0~,
\]
as the initial vector $\x(0)$ is empty, so $\Phi(\x(0))=0$.
Finally, as $\Phi_\tau=\Phi(\x(\tau-1))+2s$, and $\x(\tau-1)$ has all positive coordinates, and it is the vector of normalized distances to ${\cal K}_{|C_{\tau-1}}$ which by our assumption has size at most $s$, we have that $\Phi(\x(\tau-1))\le s$, and thus $\E[\Phi_\tau]\le 3s$.
Finally, we obtain that
\[
\E[|\threate_z|]=\E[j_\tau]\le 3s/\zeta = 6e^{(2s+1)\cdot(u'+\gamma)/u}~.
\]
This completes the proof of \lemmaref{lem:threat}.
\end{proof}

\subsection{Analysis of the Cutting Process: Bounding the Probability of Cutting a Ball}
\label{sec:cut}

In this section we give a bound on the probability that a ball is cut by a cutting skeleton-process, which depends on the {\em expected} number of threateners.

Consider a cutting skeleton-process as in \defref{def:pro} with parameters $0\le l<u\le 1$, $b>0$. Fix $z\in V(G)$, a parameter $\gamma>0$ and set $B_z=B_G(z,\gamma\Delta)$. Let $\threat_z=\{A_i \mid d_{G_i}(z,A_i)\le (u+\gamma)\Delta\}$ be the set of vertex sets whose subset $K_i$ may intersect $B_z$. Let $N := |\threat_z|$ be a random variable with $\tau = \E[N]$. We say that $B_z$ is {\em cut} by the skeleton-process if it intersects more than a single $K_i$.
\begin{lemma}\label{lem:cut}
For $\delta=e^{-2b\gamma/(u-l)}$, the probability that $B_z=B_G(z,\gamma\Delta)$ is cut by a cutting skeleton-process with the property that $\tau=\E[|\threat_z|]$, is at most
\[
(1-\delta)\left(1+\frac{\tau}{e^b-1}\right)~.
\]
\end{lemma}

Let us introduce some more notation and properties before proving this Lemma.
Define the following events:
\begin{alignat*}{2}
\cC_i &= \{B_z\cap K_i\notin\{\emptyset,B_z\}\} & \qquad &
\text{``$B_z$ cut in round $i$ (by $A_i$)''}, \\
\cF_i &= \{B_z\cap K_i=\emptyset\} & & \text{``$i$ was a \texttt{no-op}
  round''}, \\
\cE_i &= \bigg\{\cC_i\wedge\bigwedge_{j<i}\cF_j\bigg\} & &
\text{``$B_z$ first cut in round $i$''.}
\end{alignat*}

Denote by $\cF_{(< i)}$ the event $\bigwedge_{j<i}\cF_j$, so that $\cE_i =
(\cC_i \land \cF_{(<i)})$.
Denote by $\cFbar_i$ the complement of $\cF_i$. Observe that $\cC_i$ (respectively $\cFbar_i$) implies that $A_i\in\threat_z$, so
\begin{eqnarray}\label{eq:full1}
\Pr[\cC_i]=\Pr[\cC_i\wedge A_i\in\threat_z]=\Pr[A_i\in\threat_z]\cdot\Pr[\cC_i\mid A_i\in\threat_z]~,\\\label{eq:full2}
\Pr[\cFbar_i]=\Pr[A_i\in\threat_z]\cdot\Pr[\cFbar_i\mid A_i\in\threat_z]~.
\end{eqnarray}
and the same holds also when conditioning on any other event. We have the following claim:

\begin{claim}\label{claim:cut}
For each $i\in\N$,
\[
\Pr[\cC_i\mid \cF_{(<i)}, A_i\in\threat_z]\le (1-\delta)\cdot\left(\Pr[\cFbar_i\mid \cF_{(<i)},A_i\in \threat_z]+\frac{1}{e^b-1}\right)~.
\]
\end{claim}
\begin{proof}
Fix any graph $G_i$ and any set $A_i\subseteq V(G_i)$ that agree with the conditioning on $\cF_0,\dots \cF_{i-1}$ and so that $A_i\in \threat_z$. Denote by $\rho=d_{G_i}(A_i,B_z)$, $\bar{b}=b/(u-l)$, and let $m=\max\{l,\rho\}$ . Recall that $R_i$ is chosen independently, so
\begin{eqnarray*}
\Pr[\cFbar_i\mid \cF_0,\dots,\cF_{i-1},A_i\in \threat_z,A_i]&=&\int_m^u\frac{\bar{b} e^{-\bar{b}y}}{e^{-\bar{b}l}-e^{-\bar{b}u}}dy\\
&=&\frac{e^{-\bar{b}m}-e^{-\bar{b}u}}{e^{-\bar{b}l}-e^{-\bar{b}u}}~.
\end{eqnarray*}
Since $\cF_0,\dots,\cF_{i-1}$ occurred and $G_i\supseteq G_0\setminus(\cup_{j<i}K_j)$, we have that $B_z\subseteq
G_i$. Now if $R_i\ge\rho+2\gamma$ then by the triangle inequality
$B_z\subseteq K_i$, and the ball is ``saved''. This bounds the cut
probability thus:
\begin{eqnarray*}
\Pr[\cC_i\mid \cF_{(<i)},A_i\in \threat_z,A_i]&\le&\int_m^{\rho+2\gamma}\frac{\bar{b} e^{-\bar{b}y}}{e^{-\bar{b}l}-e^{-\bar{b}u}}dy\\
&\le&\frac{e^{-\bar{b}m}-e^{-\bar{b}(m+2\gamma)}}{e^{-\bar{b}l}-e^{-\bar{b}u}}\\
&=&\frac{e^{-\bar{b}m}(1-\delta)}{e^{-\bar{b}l}-e^{-\bar{b}u}}\\
&=&(1-\delta)\cdot\Pr[\cFbar_i\mid \cF_{(<i)},A_i\in \threat_z,A_i]+(1-\delta)\frac{e^{-\bar{b}u}}{e^{-\bar{b}l}-e^{-\bar{b}u}}\\
&=& (1-\delta)\cdot\left(\Pr[\cFbar_i\mid \cF_{(<i)},A_i\in \threat_z,A_i]+\frac{1}{e^b-1}\right)~.
\end{eqnarray*}
Finally, because the bound holds for any $A_i$, it holds without
conditioning on it.
\end{proof}

\begin{proof}[Proof of \lemmaref{lem:cut}]

Observe that for each $i\in[N]$, the events $\left\{\cFbar_i\wedge\cF_{(<i)}\right\}$ are pairwise disjoint (this is the event that $B_z$ is either cut or contained in $K_i$ for the first time),
thus by the law of total probability,
\begin{equation}\label{eq:gq}
\sum_{i\in\N}\Pr\left[\cFbar_i\wedge\cF_{(<i)}\right]\le 1~.
\end{equation}
Also, by linearity of expectation
\begin{equation}\label{eq:lin}
\tau = \sum_{i\in\N}\Pr[A_i\in\threat_z]~.
\end{equation}
To bound the probability of the ball being cut, we start off with the
trivial union bound:
\begin{align*}
\Pr\left[\bigcup_{i\in\N}\cE_i\right]&\le \sum_i\Pr[\cE_i] = \sum_i\Pr\left[\cC_i\wedge\cF_{(<i)}\right]\\
&=\sum_i\Pr[\cC_i\mid \cF_{(<i)}] \cdot \Pr[\cF_{(<i)}] \\
&\stackrel{\eqref{eq:full1}}{=}
  \sum_i \Pr[\cC_i \mid \cF_{(<i)}, A_i\in\threat_z]
  \cdot \Pr[A_i\in\threat_z\mid \cF_{(<i)}]
  \cdot \Pr[\cF_{(<i)}] \\
&\stackrel{\text{\footnotesize \claimref{claim:cut}}}{\leq}
  \sum_i (1-\delta)\left(\Pr[\cFbar_i\mid \cF_{(<i)},A_i\in\threat_z]+\frac{1}{e^b-1}\right)
  \cdot \Pr[A_i\in\threat_z\mid \cF_{(<i)}]
  \cdot \Pr[\cF_{(<i)}] \\
 &\stackrel{\eqref{eq:full2}}{=} (1-\delta)\cdot\sum_i\Pr\left[\cFbar_i\wedge\cF_{(<i)}\right]+\sum_i\Pr\left[A_i\in\threat_z\wedge\cF_{(<i)}\right]\cdot\frac{1-\delta}{e^b-1}\\
 &\stackrel{\eqref{eq:gq}}{\le}(1-\delta)+\frac{1-\delta}{e^b-1}\cdot\sum_i\Pr[A_i\in\threat_z]\\
 &\stackrel{\eqref{eq:lin}}{=}(1-\delta)\left(1+\frac{\tau}{e^b-1}\right)~.
\end{align*}
This completes the proof.
\end{proof}

\section{A Weak-Diameter Partition}
\label{sec:weak-diam}

In this section, we show how to construct a weak-diameter partition for
$K_{r+1}$-minor-free graphs which is $O(r)$-padded (with constant $\delta=1/160$). The ideas here will later
extend to the case of strong-diameter partitions with a weaker
$(O(r^2),O(1/r^2))$-padding.

\subsection{The Algorithm}
\label{sec:weak-algo}

At a high level, the algorithm works as follows: in each step, pick a
connected component of the remaining graph, and find (in a specific way)
a shortest-path tree $T$ in this component. Delete a random neighborhood
of $T$ from the graph, and recurse on each connected component of the
graph, if any. We then construct a net of points on each
tree, and from these net points grow ``balls'' of random radius to form
the small-diameter regions of the partition. A key property to ensure
the padding guarantee is that each vertex is expected to be close to few
of these paths. We show that this property holds, otherwise we can
construct a $K_{r+1}$-minor in $G$.

More specifically, the algorithm maintains a set of trees $T_i$ and
supernodes $S_i$ that will be used in the construction, each tree and
supernode have a ``center'' vertex associated with them. Let us
describe a generic $i$-th iteration of the algorithm. Let ${\cal S}$ be
the set containing all the supernodes created so far, initially this
will be empty. Let $C$ be a connected component in the graph
$G_i=G\setminus(\cup{\cal S})$, where $\cup \S$ is the set of all
vertices lying in the supernodes in $\S$, initially $G\setminus(\cup{\cal S})$ will be the
entire graph. Let ${\cal S}_{|C} = \{S\in{\cal S}: S\sim C\}$ be the set
of supernodes that have a neighbor in component $C$. Say $\SC = \{S'_1,
S'_2, \ldots, S'_k\}$, and consider the vertices $F_j = N(S'_j) \cap C$ for
each supernode, which are vertices in $C$ neighbors of these
``adjacent'' supernodes. (These $F_j$'s may intersect.) We pick an
arbitrary vertex $u_i$ from $C$ and build a tree $T_i$ rooted at $u_i$,
which is comprised of shortest paths from $u_i$ to each of the sets
$F_j$ (that is, for each $j$ take a shortest path from $u_i$ to the nearest vertex in $F_j$). Define the next supernode
\[ S_i := B_{G_i}(T_i,R_i\Delta), \] where $R_i
\propto \Texp_{[0,1/8]}(16r)$. (Recall the definition of the truncated exponential distribution from \eqref{eq:texpgen}.)

In order to create the random partition, choose a $\Delta/8$-net $N_i$
over $T_i$, and enumerate $N_i=\{v_1,\dots,v_{|N_i|}\}$. For each $1\le
j\le |N_i|$, create a cluster $B_{G_i}(v_j,\alpha_j\Delta)\cap U_j$
(where $U_j$ is the set of points which have no cluster yet), where each
$\alpha_j \propto \Texp_{[1/4, 1/2]}(160r)$. This completes the description
of the algorithm; it is also given as Algorithm~1 and~2.

\begin{algorithm}\label{alg:a}
\caption{\texttt{Weak-Random-Partition}($G$,$\Delta$,$r$)}
\begin{algorithmic}[1]
\STATE Let $G_0 \leftarrow G$, $i\leftarrow 0$.
\STATE Let $\mathcal{S} \leftarrow \emptyset$.
\STATE Let $\mathcal{T} \leftarrow \emptyset$.
\WHILE {$G_i$ is non-empty}
\STATE Let $C_i$ be a connected component of $G_i$.
\STATE Pick $u_i \in C_i$. Let $T_i$ be a tree rooted at $u_i$ that consists of shortest paths (in $G_i$) from $u_i$ to the closest vertex of $N(S)$ for each supernode $S\in\mathcal{S}_{|C_i}$.
\STATE Let $R_i$ be a random variable drawn independently from the
distribution $\Texp_{[0,1/8]}(16r)$.
\STATE Let $S_i \leftarrow B_{G_i}(T_i,R_i\Delta)$ be a neighborhood of $T_i$.
\STATE Add $S_i$ to $\mathcal{S}$.
\STATE Add $T_i$ to $\mathcal{T}$.
\STATE $G_{i+1} \leftarrow G_i \setminus S_i$.
\STATE $i\leftarrow i+1$.
\ENDWHILE
\RETURN \texttt{Create-Balls}($G$,${\cal T}$,${\cal S}$,$\Delta$,$r$).
\end{algorithmic}
\end{algorithm}

\begin{algorithm}\label{alg:b}
\caption{\texttt{Create-Balls}($G$,${\cal T}$,${\cal S}$,$\Delta$,$r$)}
\begin{algorithmic}[1]
\STATE $P=\emptyset$.
\FOR {$i=1,\dots,|{\cal T}|$}
\STATE Let $N_i=\{v_1,\dots,v_{|N_i|}\}$ be a $\Delta/8$-net of $T_i$.
\FOR {$j=1,\dots,|N_i|$}
\STATE Let $\alpha_j$ be a random variable drawn independently from
distribution $\Texp_{[1/4,1/2]}(80r)$
\STATE Add $B_{G_i}(v_j,\alpha_j\Delta)\setminus\cup P$ as a cluster to the partition $P$.
\ENDFOR
\ENDFOR
\RETURN $P$.
\end{algorithmic}
\end{algorithm}

\begin{figure}[ht]
\begin{center}
    \fbox{\includegraphics[width=143pt]{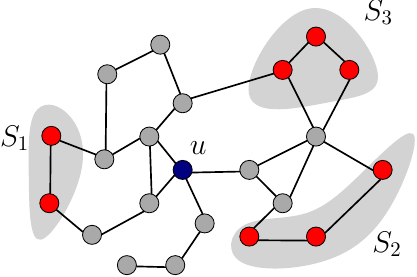}}
    \hspace{1pt}
    \fbox{\includegraphics[width=143pt]{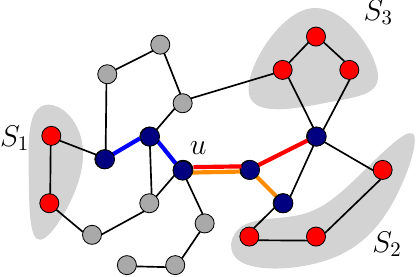}}
    \hspace{1pt}
    \fbox{\includegraphics[width=143pt]{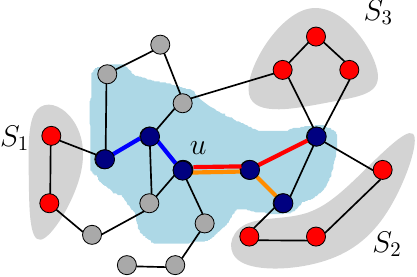}}
    \caption{\emph{An iteration of the algorithm. On the left, there are
        three supernodes $S_1,S_2,S_3$ neighboring the current component
        with $u$ as a root. In the middle, we have a tree $T_4$
        comprised of three shortest path from $u$. On the right, the new supernode $S_4$ which is a $1$-neighborhood of $T_4$ (observe that this neighborhood is taken in the connected component containing $u$).}}
        \end{center}
\end{figure}

\subsection{The Analysis}
\label{sec:weak-analysis}

The following invariant holds for each time step $i$:
\begin{invariant}\label{inv:sim1}
  For every $i\ge 0$, every connected component $C$ of $G_{i}$ satisfies
  that if $S,S'\in{\cal S}_{|C}$ then $S\sim S'$.
\end{invariant}

\begin{proof}
  The proof is by induction; the base case is trivial as there are no
  supernodes in $\SC$. Now by induction, assume that the invariant holds
  in $G_{i}$. Let $T_{i}$ and $S_{i}$ be the tree and supernode
  constructed in step $i$ in the component $C_{i}$. Let $C$ be some
  connected component of $G_{i+1}$, and $S,S'\in{\cal S}_{|C}$. If
  $C\cap C_{i}=\emptyset$ then $C$ is a component of $G_{i}$ as well;
  moreover, as $S_i\subseteq C_i$ it must be that $S_{i}\nsim C$ so neither of $S,S'$ can be $S_i$, and hence
  we can use the induction hypothesis to infer that $S\sim S'$. On the
  other hand, suppose that $C\subseteq C_{i}$. There are two cases: if
  $S_{i}\notin\{S,S'\}$ we have $S\sim S'$ by the induction hypothesis
  on $C_{i}$. On the other hand, suppose $S_{i} = S$ (w.l.o.g.). Recall
  that $T_{i}$ was chosen so that it contains a neighbor of every supernode in ${\cal
    S}_{|C_{i}}$ and $T_{i}\subseteq S_{i}$, we have that $S_{i}\sim
  S'$.
\end{proof}
\invref{inv:sim1} implies that for each connected component $C$,
contracting the supernodes of ${\cal S}_{|C}$ yields a $K_{|{\cal
    S}_{|C}|}$ minor, so we obtain the following corollary.
\begin{corollary}\label{cor:ler1}
  If $G$ excludes $K_{r+1}$ as a minor, then for every time step $i$,
  the connected component $C_{i}$ has $|{\cal S}_{|C_{i}}|\le r$. In
  particular, the tree $T_i$ is made up of at most $r$ shortest paths
  in $G_i$.
\end{corollary}

\begin{claim}
  \label{clm:weak-diam}
  The algorithm above generates a $\Delta$-bounded partition of $G$.
\end{claim}

\begin{proof}
  First we prove that we generate a partition. Indeed, we delete
  supernodes from the graph, and recurse on the remaining components, so
  we need to show that vertices within the supernodes are contained in
  some cluster. Consider a vertex $x$ in supernode $S_i$. By definition,
  $d_{G_i}(x,T_i) \leq \Delta/8$. Since $N_i$ is a $\Delta/8$-net in
  $T_i$, some net point $v_j \in N_i$ satisfies $d_{G_i}(x,v_j) \leq
  \Delta/4$. And since $\alpha_j \geq 1/4$, the ball $B_{G_i}(v_j, \alpha_j
  \Delta)$ contains $x$. Hence each vertex within the deleted supernode is
  contained in some cluster, and we get a partition of $G$.  Moreover,
  each cluster is a ball of radius at most $\alpha_j \Delta \leq \Delta/2$
  (and hence diameter at most $\Delta$) in $G_i$. Finally, distances in
  $G_i$ are no smaller than those in $G$.
\end{proof}

\begin{lemma}
  \label{lem:weak-cut-prob}
  For $r \geq 4$, and any $\gamma\le 1/160$, the probability that a ball $B_z$ of radius
  $\gamma\Delta$ is cut by the above process is
  \[ \Pr[ B_z~ \mathrm{cut}] \leq 1 - e^{-320r\gamma}. \]
\end{lemma}
\begin{proof}
  First observe that the process defined in Algorithm 1 is a threatening skeleton-process, with the sequence of
  graphs $G_0,G_1,\dots$ as defined in the algorithm and with $A_i=T_i$, $K_i=S_i$, $l=0, u=1/8$, $s=r$ and $b=2s$. Recall that
  $B_z=B_G(z,\gamma\Delta)$, and set $u'=1/2$ so that $\threate_z=\{T_i \mid d_{G_i}(z,T_i)\le
  (u'+\gamma)\Delta\}$ (we choose this $u'$ to accommodate the cutting process which will be conducted with this parameter). By \invref{inv:sim1} we get that for all $i\in\N$, $|{\cal
    S}_{|C_i}|\le r$, so by \lemmaref{lem:threat} (using that
  $\gamma\le 1/160$),
  \begin{equation}\label{eq:thr}
    \E[|\threate_z|]\le 6e^{(2r+1)\cdot(u'+\gamma)/u}\le 10e^{10r}~.
  \end{equation}
  For each $i$ such that $T_i\in\threate_z$, let $U_i=\{v\in N_i \mid
  d_{G_i}(v,z)\le (1/2+\gamma)\Delta\}$ be the net points in $N_i$ that are
  sufficiently close to threaten $B_z$ (note that by the choice of $u'$ this is indeed the case), and denote $\threat_z=\cup_{i \mid T_i\in\threate_z}  U_i$. By \corollaryref{cor:ler1}, $T_i$ is comprised of at most $r$
  shortest paths, and we claim that on each shortest path there can be
  at most $10$ points that are in $U_i$. This is because the distance
  between any two consecutive net points on a path is at least
  $\Delta/8$, and if there are $q>10$ points, because this is a shortest
  path, the distance from the first point to the last is at least
  $(q-1)\cdot\Delta/8>(1+2\gamma)\Delta$. The triangle inequality
  implies that it can't be that both are within $(1/2+\gamma)\Delta$
  from $z$. We conclude that for all
  $i$ (with $T_i\in\threate_z$) we have $|U_i|\le 10r$, thus by
  \eqref{eq:thr}
  \begin{equation}\label{eq:ut}
    \tau:=\E[|\threat_z|]\le 10r\cdot 10e^{10r}=100r\cdot e^{10r}~.
  \end{equation}
  Next, we show that our \texttt{Create-Balls} algorithm generates a cutting skeleton-process. Simply
  take the sequence $G_0,\dots,G_0,G_1,\dots,G_1,G_2,\dots$, where each
  $G_i$ is taken $|N_i|$ times. Then the skeleton sets $A$ are in fact
  singletons: for each $i$ we will take $|N_i|$ sets - the points of
  $N_i$, to be these singletons. The parameters for the exponential
  distribution are $l=1/4$, $u=1/2$ and $b=20r$. To see the cutting property of \defref{def:pro}, note that once we move from the graph $G_i$ to $G_{i+1}$, $G_{i+1}$ will contain all the
  points yet uncovered by clusters, because we already observed in \claimref{clm:weak-diam}
  that once all the points of $N_i$ create a cluster, the supernode
  $S_i$ is completely covered (recall $G_{i+1}=G_i\setminus
  S_i$). Finally, applying \lemmaref{lem:cut}, we obtain that the
  probability that $B_z$ is cut is at most
  \[
  (1-e^{-2b\gamma/(u-l)})\!\left(1\!+\!\frac{\tau}{e^b-1}\right) =
  (1-e^{-160r\gamma})\!\left(1\!+\!\frac{100r \cdot e^{10r}}{e^{20r}-1}\right).
  \]
  It holds that $\frac{100r \cdot e^{10r}}{e^{20r}-1} \leq e^{-r}$ for $r \geq
  4$, and this completes the proof as
  \[
  (1-e^{-160r\gamma})\cdot(1+e^{-r})\le(1-e^{-160r\gamma})\cdot(1+e^{-160r\gamma})= 1-e^{-320r\gamma}~,
   \]
   using that $\gamma\le 1/160$.
\end{proof}

\section{A Strong-Diameter Partition}
\label{sec:strong-diam}

In the previous section, we saw how to get a weak-diameter partition for
minor-free graphs. In this section, we give a strong-diameter guarantee
with a slightly weaker padding parameter of $(O(r^2),O(1/r^2))$ instead of
$O(r)$. However, this is still an exponential improvement over the best
previous padding for such strong-diameter partitions of minor-free graphs.

\subsection{The Algorithm}
\label{sec:strong-algo}

The algorithm for strong-diameter partitions is similar in spirit to
that of \sectionref{sec:weak-algo} for weak-diameter partitions, but
there are some crucial differences that we highlight here.

At a high level, the algorithm works as follows: in each step, pick a
connected component of the remaining graph, and find (in a specific way)
a shortest path $P$ in this component. Delete a random neighborhood of
$P$ from the graph, and recurse on each connected component of the
graph, if any. Each such random neighborhood is decomposed into small
diameter regions using cones centered at some of $P$'s points. A key
property to ensure the padding guarantee is that each node is expected
to be close to few of these paths. We show that this property holds,
otherwise we can construct a $K_{r+1}$-minor in $G$.

The algorithm again maintains a set of paths (instead of trees), and
associated supernodes that will be used in the construction. These will
be denoted as $P_{ij}$ and $S_{i}$ respectively, and supernode $S_i$
consists of the union of neighborhoods of the paths
$P_{ij}$. The main difference from the weak-diameter construction is
that instead of building a shortest-path tree all at once, we build a
``tree'' one path at a time, and remove a neighborhood of the path from
the graph before constructing the subsequent paths.

Let us describe the $i$-th iteration of the algorithm. Let ${\cal
  S}\subseteq V$ be the set containing all the supernodes created so
far. Let $C_i$ be a connected component in the graph
$G_i=G\setminus(\cup{\cal S})$. Let ${\cal S}_{|C_i} = \{S\in{\cal S}:
S\sim C_i\}$ be the set of supernodes that have a neighbor in component
$C_i$. We pick an arbitrary vertex $u_i$ from $C_i$ and build a
supernode $S_i$. Again, the intuition behind the construction is that we
wish for the new supernode to ``touch'' every supernode $S\in{\cal
  S}_{|C_i}$ (i.e., $S_i\sim S$). However, this is done slightly
differently from \sectionref{sec:weak-algo}. At the
first iteration ($j=1$) we create a shortest path $P_{ij}$ from $u_i$ to
some supernode $S\in{\cal S}_{|C_i}$, and remove a random neighborhood
$S_{ij}$ from the graph to obtain $G_{i(j+1)}$.  This neighborhood
$S_{ij}$ is defined as all the vertices within distance
$R_{ij}\cdot\Delta$ of $P_{ij}$ (in the current component $C_{ij}$),
where $R_{ij}\propto\Texp_{[0,1/4]}(8(r^2+r))$. We increase the iteration
counter $j$ and continue in this manner on every connected component of
$G_{ij}$ that is contained in $C_i$, until the new supernode
$S_i=\cup_jS_{ij}$ touches every supernode $S\in{\cal S}_{|C}$ for every
connected component $C\subseteq C_i$ in the remaining graph
$G_{ij}$.

Finally, each such neighborhood $S_{ij}$ is partitioned to
``cones''. Each cone $B$, centered at some (yet uncovered) point $c\in
P_{ij}$, consists of the (yet uncovered) points in $S_{ij}$ whose
distance to $c$ is not ``much larger'' than their distance to $P_{ij}$.
The notion of being ``much larger'' is determined by a random variable
$\alpha$ drawn independently and uniformly from $[\Delta/8,\Delta/4]$.
The algorithms are formally presented as Algorithms~3 and~4
respectively. Observe that the subroutine \texttt{Create-Cones} is
invoked in line~13 of \texttt{Strong-Random-Partition}.

\begin{algorithm}\label{alg:aprime}
\caption{\texttt{Strong-Random-Partition}($G$,$\Delta$,$r$)}
\begin{algorithmic}[1]
\STATE Let $G_0 \leftarrow G$, $i\leftarrow 0$.
\STATE Let $\mathcal{S} \leftarrow \emptyset$.
\STATE Let $\mathcal{C} \leftarrow \emptyset$.
\WHILE {$G_i$ is non-empty}
\STATE Select a connected component $C_i$ of $G_i$, and pick $u_i\in C_i$.
\STATE Let $W=\{u_i\}$.
\STATE Let $j=1$ and $G_{ij}=G_i\setminus W$.
\WHILE {there exist a connected component $C_{ij}$ in $G_{ij}$ and a supernode $S\in{\cal S}_{|C_{ij}}$ such that $C_{ij}\sim S$ and $C_{ij}\sim W$ but $W\nsim S$}
\STATE Choose $u\in N(W)\cap C_{ij}$.
\STATE Let $P_{ij}$ be a shortest path (in $G_{ij}$) from $u$ to $N(S)$.
\STATE Let $R_{ij}$ be a random variable drawn independently from the distribution $\Texp_{[0,1/4]}(8(r^2+r))$.
\STATE Let $S_{ij} \leftarrow B_{G_{ij}}(P_{ij},R_{ij}\Delta)$ be a neighborhood of $P_{ij}$.
\STATE \texttt{Create-Cones}($S_{ij}$,$P_{ij}$,$\cC$).
\STATE $W\leftarrow W\cup S_{ij}$.
\STATE $G_{i(j+1)} \leftarrow G_{ij}\setminus S_{ij}$.
\STATE $j\leftarrow j+1$.
\ENDWHILE
\STATE Set $S_i=W$, and add $S_i$ to ${\cal S}$.
\STATE $G_{i+1} \leftarrow G_i \setminus S_i$.
\STATE $i\leftarrow i+1$.
\ENDWHILE
\RETURN ${\cal C}$
\end{algorithmic}
\end{algorithm}

\begin{algorithm}\label{alg:cone}
\caption{\texttt{Create-Cones}($S$,$P$,${\cal C}$)}
\begin{algorithmic}[1]
\WHILE {$P\neq \emptyset$}
\STATE Choose $c\in P$.
\STATE Choose $\alpha\in[1/8,1/4]$ uniformly at random.
\STATE Let $B=\{u\in S \mid d_S(u,c)-d_S(u,P)\le \alpha\Delta\}$.
\STATE Add $B$ to $\cC$.
\STATE Set $S\leftarrow S \setminus B$.
\STATE Set $P\leftarrow P \setminus B$.
\ENDWHILE

\end{algorithmic}
\end{algorithm}

\subsection{The Analysis}

We begin by arguing that the algorithm creates a partition ${\cal C}$ with strong-diameter $\Delta$. The following properties will be useful.
\begin{proposition}\label{prop:cone}
For any $S$ and $P$ obtained during the run of the algorithm \texttt{Create-Cones}:
\begin{itemize}
\item If $u,v\in S$ are such that a shortest path from $u$ to $P$ contains $v$, and $v\in B$ for a cone $B$, then also $u\in B$.

\item If $u,v\in S$ are such that a shortest path from $u$ to $c$ contains $v$, and $u\in B$ for a cone $B$ centered at $c$, then also $v\in B$.
\end{itemize}

\end{proposition}
\begin{proof}
Let $c\in P$ be the center of the cone $B$. We begin by proving the first item: Since $v\in B$ we have that $d_S(v,c)-d_S(v,P)\le \alpha\Delta$. Since $v$ is on the shortest path from $u$ to $P$, $d_S(u,P)=d_S(u,v) +d_S(v,P)$ and thus
\begin{eqnarray*}
d_S(u,c)-d_S(u,P) &\le& (d_S(u,v)+d_S(v,c))-(d_S(u,v) +d_S(v,P))\\
&=& d_S(v,c)-d_S(v,P)\\
&\le& \alpha\Delta~,
\end{eqnarray*}
which implies that $u\in B$.

The second item is proved in a similar manner: Since $u\in B$ we have that $d_S(u,c)-d_S(u,P)\le \alpha\Delta$. Since $v$ is on the shortest path from $u$ to $c$, $d_S(v,c)=d_S(u,c)-d_S(u,v)$ and thus
\begin{eqnarray*}
d_S(v,c)-d_S(v,P) &\le& (d_S(u,c)-d_S(u,v))-(d_S(u,P) -d_S(u,v))\\
&=& d_S(u,c)-d_S(u,P)\\
&\le& \alpha\Delta~,
\end{eqnarray*}
which implies that $v\in B$.
\end{proof}

\begin{lemma}\label{lem:strong}
Each cone $B$ created in the algorithm has $\diam(G[B])\le\Delta$.
\end{lemma}
\begin{proof}
Recall that each neighborhood $S$ of a shortest path $P$ contains points within distance at most $\Delta/4$ from $P$. Let $S$ be the remaining part after some cones have been created, and $P$ is the remaining path. The first property in
\propref{prop:cone} implies that the shortest path from any $u\in S$ to $P$ is fully contained in $S$, and thus
\begin{equation}\label{eq:hf}
d_{S}(u,P)\le \Delta/4~.
\end{equation}
Consider a certain cone $B$ centered at $c\in P$, and by definition of $B$, for each $u\in B$,
\begin{equation}\label{eq:dd}
d_{S}(u,c)\le \alpha\Delta+d_{S}(u,P)\stackrel{\eqref{eq:hf}}{\le} \Delta/4+\Delta/4=\Delta/2~.
\end{equation}
By the second property of \propref{prop:cone}, if $u\in B$ then surely any $v\in S$ on the shortest path from $u$ to $c$ will also be in $B$, so $d_B(u,c)\le\Delta/2$ as well, and thus $\diam(G[B])\le\Delta$.

Finally, it remains to see that \texttt{Create-Cones} generates a partition of $S$ (i.e. that the clusters it creates cover $S$), and this can be verified by the first property of \propref{prop:cone}. If for $u\in S$ there is a shortest path from $u$ to $P$ ending at $v\in P$, then whenever $v$ is covered by a cone, $u$ must be covered as well (the algorithm does not stop until $P=\emptyset$).
\end{proof}

For a time step $i$, we say that $W$ is the {\em working supernode}, and at the end of this step it will become the supernode $S_i$. Note that $W$ induces a connected subgraph, because we always choose a vertex $u$ in $N(W)$ to be a start of the next path. We denote by $G_{i0}=G_i$. The following invariant holds for each time step $i$:
\begin{invariant}\label{inv:sim}
For every $i,j\ge 0$, every connected component $C$ of $G_{ij}$ satisfies that if $S,S'\in{\cal S}_{|C}$ then $S\sim S'$.
\end{invariant}
\begin{proof}
Assume inductively that the invariant holds until time step $i$ at iteration $j$. First consider the case $j>0$, then as $G_{ij}$ is obtained from $G_{i(j-1)}$ by removing some vertices, and the set of supernodes remains unchanged, the invariant will still hold: Every connected component $C$ of $G_{ij}$ is a subset of a connected component $D$ of $G_{i(j-1)}$, in particular ${\cal S}_{|C}\subseteq {\cal S}_{|D}$, and so any pair of supernodes $S,S'\in{\cal S}_{|C}$ is also in ${\cal S}_{|D}$ and thus $S\sim S'$.

For the case $j=0$, a new supernode $S_{i-1}$ was just introduced, but the termination condition of line 9. guarantees that for any connected component $C$ in $G_i$, any supernode $S\in{\cal S}_{|C}$ must have $S\sim S_{i-1}$.
\end{proof}

\begin{corollary}\label{cor:ler}
If $G$ excludes $K_{r+1}$ as a minor, then for every time step $i$ and iteration $j$, the connected component $C_{ij}$ has $|{\cal S}_{|C_{ij}}|\le r$. Moreover, fix some $z\in V$. If $P_{i1},\dots,P_{il}$ are the shortest paths chosen while creating $S_i$ in the components containing $z$, then $l\le r$.
\end{corollary}
\begin{proof}
If $|{\cal S}_{|C_{ij}}|=q$, then using \invref{inv:sim}, contracting each supernode in ${\cal S}_{|C_{ij}}$ will yield a $K_q$ minor, so it must be that $q\le r$. To see the second part of the assertion, note that each $P_{ij}$ will connect the component containing $z$ with some supernode $S\in {\cal S}_{|C_{ij}}$, so that $S_{ij}\sim S$. Finally, as $|{\cal S}_{|C_{ij}}|\le r$, there can be at most $r$ such paths.
\end{proof}

\begin{lemma}
  \label{lem:strong-cut-prob}
  For $\gamma\le 1/r^2$, the probability that a ball $B_z$ of radius
  $\gamma\Delta$ is cut by the above process is
  \[ \Pr[ B_z \mbox{ is cut}] \leq O(\gamma r^2)~. \]
\end{lemma}

\begin{proof}
  First observe that our algorithm is a threatening skeleton-process with parameters
  $l=0$, $u=1/4$, $s=r^2+r$, $b=2s$ and the $G_i$ (respectively $A_i$, $K_i$) are the
  $G_{ij}$ (resp. $P_{ij}$, $S_{ij}$) ordered lexicographically. By
  \invref{inv:sim} we get that for all $i,j\in\N$, $|{\cal
    S}_{|C_{ij}}|\le r$. By \corollaryref{cor:ler}, each of these
  supernodes $S\in{\cal S}_{|C_{ij}}$ can have at most $r$ paths that
  were built in a component containing $C_{ij}$, so it may contribute at
  most $r$ to the number of sets in ${\cal K}_{|C_{ij}}$, to a total of
  $r^2$. We must also add in the (at most) $r$ paths of the current
  working supernode, to obtain that $|{\cal K}_{|C_{ij}}|\le s$. For $u'=u$, we set $\threat_z=\{P_{ij}  \mid d_{G_{ij}}(P_{ij},z)\le(u+\gamma)\Delta\}$, and
  let $\tau=\E[|\threat_z|]$. With this we may apply
  \lemmaref{lem:threat} to infer that
  \[
  \tau\le 6e^{(2s+1)\cdot(u'+\gamma)/u}=6e^{(2s+1)\cdot(1+\gamma/u)}~.
\]
  Next, we show that our process is also a cutting skeleton-process, with the graph sequence $G_{ij}$ and
  the skeletons are the $P_{ij}$, ordered lexicographically. The
  parameters are the same as before: $l=0$, $u=1/4$ and $b=2s$ (this is
  the exact same process, after all). The condition that the graph
  sequence contains every uncovered point is trivial by definition of
  $G_{ij}$. By \lemmaref{lem:cut} we obtain that the probability that
  $B_z$ is cut is at most
  \begin{eqnarray}\label{eq:bb}
    \lefteqn{(1-e^{-2b\gamma/(u-l)})\left(1+\frac{\tau}{e^b-1}\right)}\nonumber\\&\le& (1-e^{-20r^2\gamma})\cdot(1+O(e^{10r^2\gamma}))\nonumber\\
    &=&O(\gamma r^2)~,
  \end{eqnarray}
  where the inequality uses that $s=r^2+r\le 5r^2/4$ (as $r\ge 2$), and the last equality uses that $\gamma\le 1/r^2$ and $e^x\approx 1+O(x)$ whenever $|x|\le 20$.
  In what follows we bound the probability of event $\cE_{\cone}$, which
  is the event that the ball $B_z$ is cut in the \texttt{Create-Cones}
  procedure, while conditioning that it was not cut while creating the $S_{ij}$. Let
  $S=S_{ij}$ be the set that contains $B_z$, which was built around the
  path $P=P_{ij}$.  Let $c_1,\dots,c_k$ be the centers chosen in
  \texttt{Create-Cones}($S$,$P$,$\cC$). We claim that there can be at most $9$ of them that may cut $B_z$. To see this, observe that each cone contains a ball of radius at least $\Delta/8$, and since $P$ is a shortest path, in any set of 10 centers there are two centers $c_g,c_h$ such that $d_S(c_g,c_h)\ge 9\Delta/8>2(1/2+\gamma)\Delta$. By the triangle inequality it must be that at least one of them is more than $(1/2+\gamma)\Delta$ away from $z$. Finally, by \lemmaref{lem:strong} any cone centered at $c$ may only contain points at distance at most $\Delta/2$ from $c$ (see
  \eqref{eq:dd}), so it may not be the first to cut $B_z$.
  As $\alpha$ is chosen uniformly from an interval of
  size $\Delta/8$, the probability that a ball of radius $\gamma\Delta$
  will be cut is at most $2\gamma\Delta/(\Delta/8)=16\gamma$. By a
  simple union bound,
  \[
  \Pr[\cE_{\cone}\mid B_z\subseteq S]< 144\gamma~,
  \]
  which is dominated by \eqref{eq:bb}, thus the final bound is
  \[
  \Pr[B_z\text{ is cut}]\le O(\gamma r^2)~.
  \]
\end{proof}

\section{Bounded Tree-width Graphs}\label{sec:tree}

In this section we prove the second part of \theoremref{thm:stong-list}, that any graph with tree-width at most $r$ admits an efficient $(O(r),O(1/r))$-padded strong-diameter partition scheme.

Since graphs of tree-width $r$ are $K_{r+2}$-minor-free, the result of \sectionref{sec:weak-diam} already implies a (weak diameter) probabilistic partition which is $O(r)$-padded. The purpose of this section is to show a {\em strong-diameter} $(O(r),O(1/r))$-padded partition for graphs of bounded tree-width.
We will use the same framework as the previous sections, and exploit the special structure of bounded tree-width graphs.

\begin{definition}
A graph $G=(V,E)$ has tree-width $r$ if there exists a collection of sets $I=\{X_1,\dots,X_k\}$ with each $X_i\subseteq V$, and a tree $T=(I,F)$, such that the following conditions hold:
\begin{itemize}
\item $\cup_{i\in[k]}X_i=V$,
\item For all $i\in[k]$, $|X_i|\le r+1$,
\item For all $\{u,v\}\in E$, there exists $i\in[k]$ such that $u,v\in X_i$,
\item For all $u\in V$, the tree nodes containing $u$ form a connected subtree of $T$.
\end{itemize}
\end{definition}

\begin{corollary}\label{cor:width}
Let $U$ be a bag in the tree-decomposition $T=(I,F)$ of $G=(V,E)$. Then if $U_1,U_2\in I$ lie in different connected components of $T\setminus \{U\}$, and $x_1\in U_1\setminus U$, $x_2\in U_2\setminus U$, then $x_1,x_2$ are in different connected components of $G\setminus U$.
\end{corollary}

\subsection{The Algorithm}

Let $G=(V,E)$ be a graph of tree-width $r-1$, and let $T$ be its tree-decomposition so that each bag has at most $r$ vertices, and $T$ has an arbitrary root $R$. The height of a tree node $U$, $h(U)$, is its distance in $T$ from the root $R$. For a vertex $u\in V$ let $h(v)$ denote the minimal height of a tree node $U$ containing $u$, and denote by $b(u)=U$ the node achieving this minimum. Order the vertices of the graph $(v_1,\dots v_n)$ such that for all $1\le i<j\le n$, $h(v_i)\le h(v_j)$. In the $i$-th iteration of the algorithm we will have a graph $G_i$ (initially $G_1=G$), and if $v_i\in G_i$ we shall create a cluster $S_i=B_{G_i}(v_i,R_i\Delta)$, where $R_i\propto\Texp_{[0,1/2]}(8r)$. Then set $G_{i+1}=G_i\setminus S_i$ and continue. If $v_i\notin G_i$ then we do nothing in this iteration.

\begin{algorithm}\label{alg:tree}
\caption{\texttt{Tree-width-Partition}($G$,$\Delta$,$r$)}
\begin{algorithmic}[1]
\STATE Set ${\cal S}=\emptyset$.
\STATE Let $G_1 \leftarrow G$.
\FOR {$i=1,\dots n$}
\IF {$v_i\in G_i$}
\STATE Let $R_i\propto\Texp_{[0,1/2]}(8r)$.
\STATE Let $S_i=B_{G_i}(v_i,R_i\Delta)$.
\STATE Add $S_i$ to ${\cal S}$.
\STATE Set $G_{i+1}\leftarrow G_i\setminus S_i$.
\ELSE
\STATE Set $G_{i+1}\leftarrow G_i$.
\ENDIF
\ENDFOR
\RETURN ${\cal S}$.
\end{algorithmic}
\end{algorithm}

\subsection{The Analysis}

Fix some $z\in V$, $\gamma=O(1/r)$ and $B_z=B_G(z,\gamma\Delta)$. Let $U=b(z)\in I$ be the tree node containing $z$ such that $h(z)=h(U)$. The first observation is that when analyzing the probability that $B_z$ is cut, we may restrict our attention to vertices $v\in V$ whose $b(v)$ lies on the path from $R$ to $U$ in $T$. The reason is that if $b(v_i)$ is not on this path, then if $C\in I$ is the least common ancestor of $U$ and $b(v_i)$ in $T$, we claim that $G_i$ does not contain any vertex from $C$. To see this, note that by the choice of ordering all vertices in $C$ appear before $v_i$, and thus either created a cluster or were removed from the graph. By \corollaryref{cor:width} $z$ and $v_i$ are in different component of $G_i$, so $S_i$ cannot be the first to cut $B_z$.

Consider then the process restricted to the vertices contained in bags on the path from $R$ to $U$ (we may assume w.l.o.g that these appear first in the ordering). For any $i\in[n]$, denote by $C_i$ the connected component in $G_i$ that contains $z$, and let ${\cal S}_{|C_i}=\{S_j~:~ S_j\sim C_i\}$.

\begin{claim}\label{claim:fef}
For any $i\in[n]$, $|{\cal S}_{|C_i}|\le 2r$.
\end{claim}
\begin{proof}
Let $R=U_1,\dots,U_k=U$ be the sequence of bags from the root to $U$ in the tree-decomposition. For any $j\in[k]$, let $i_j\in[n]$ be the minimal such that $U_j\cap V(G_{i_j})=\emptyset$. We prove that $|{\cal S}_{|C_{i_j}}|\le r$
, by noting that there are at most $r$ supernodes that can intersect $U_j$ (as $|U_j|\le r$). If a supernode $S_h$ does not intersect $U_j$, then since this supernode is not centered at some vertex of $U_{j'}$ for $j'>j$ (using the ordering and the minimality of $i_j$), then by \corollaryref{cor:width} there is no path from $z$ to $N(S_h)$ in $G_{i_j}$. Since there are at most $r$ new supernodes created between time $i_j$ to $i_{j+1}$ (as each bag is covered after at most $r$ clusters are formed), the claim follows.
\end{proof}

Observe that the algorithm generates a threatening skeleton-process with the sequence $G_1,\dots$, the skeletons are $A_i=\{v_i\}$, $K_i=S_i$, $l=0$, $u=1/2$, $s=2r$ and $b=4r$. Let $u'=u$ and $\threate_z=\{v_i \mid d_{G_i}(z,v_i)\le (u+\gamma)\Delta\}$. By \claimref{claim:fef} we may apply \lemmaref{lem:threat} and obtain that
\begin{equation}
\tau\le 6e^{(4r+1)\cdot(1+\gamma/u)}~.
\end{equation}
Finally, as our process can also be made to be a cutting skeleton-process, as long as we omit the steps in which $v_i\notin G_i$ (note that the next $i$ for which $v_i\in G_i$ may depend on previous random choices of $R_j$ for $j<i$, but this is allowed), and with $l=0$, $u=1/2$ and $b=4r$.
Applying \lemmaref{lem:cut}, we obtain that the probability that $B_z$ is cut is at most
\[
(1-e^{-2b\gamma})\left(1+\frac{\tau}{e^b-1}\right)\le (1-e^{-8r\gamma})\cdot O(e^{8r\gamma})
    =O(\gamma r),
\]
using that $\gamma\le 1/r$.

\subsection{Bounded Pathwidth Graphs}

A graph has path-width $r$ if it has a tree decomposition of width $r$ such that the tree is a path.
The following result was communicated to us by James R. Lee and Anastasios Sidiropoulos.
\begin{itemize}
\item Any graph $G$ on $n$ vertices and path-width $r$ admits an efficient $O(\log r)$-padded strong-diameter partition scheme.
\end{itemize}

We provide a sketch of the proof. First decompose the graph into shortest paths as follows: as long as the graph is not empty, in each connected component, pick a shortest path between a vertex in the first bag to a vertex in the last bag. Remove this path from the graph, and continue on the connected components that remain. Since any such path must use some vertex in every bag, it follows that the path-width decreases by at least 1 in each iteration. We thus obtain a cop-decomposition of width 1 and depth $r$. We now apply our method, and the number of threateners is only $O(r)$, which implies the result

Note that every graph $G$ on $n$ vertices and tree-width $r$ has path-width at most $O(r\log n)$ (this follows because it has a tree-decomposition of depth $O(\log n)$, see e.g. \cite{GTW13}). An immediate corollary is a $O(\log r+\log\log n)$-padded strong-diameter partition for graphs of tree-width $r$.

\section{Bounded Euler-Genus Graphs}\label{sec:genus}

In this section we prove the third part of
\theoremref{thm:stong-list}, that any graph with Euler-genus at most
$g$ admits an efficient $O(\log g)$-padded strong-diameter partition
scheme. We assume here that the graph $G$ is embedded without any edge
crossing on some closed surface $\Sigma$ (compact and without
boundary), which can be orientable or non-orientable, of Euler
characteristic $2-g$.

The Euler characteristic is the value $\chi(\Sigma) = n - e + f$ where
$n,e,f$ are respectively the number of nodes, edges, and faces of the
embedding of $G$ on $\Sigma$. If $\Sigma$ is orientable then $g$ must
be even and $\Sigma$ homeomorphic to a sphere with $g/2$
``handles''. And if $\Sigma$ is non-orientable, then it is
homeomorphic to a sphere with $g$ ``cross-caps''. The Euler-genus of
$G$ is the smallest $g$ such that it can be embedded on a surface of
Euler characteristic $2-g$. So, it generalizes the classical notion of
genus of a graph (for orientable surfaces) and the non-orientable
genus of a graph. Planar graphs have Euler-genus~$0$.

Using the Fundamental Cycle Method based on BFS trees
(see~\cite[Lemma~4.2.4 and Theorem 4.3.2]{MT01}), we have the
following lemma (see also~\cite{IndykS07,Colin10}):

\begin{lemma}\label{lem:genus}
  If $G$ is a Euler-genus $g$ graph, there exists a cycle $A$
  comprised of two shortest paths emanating at a common root, such
  that $G\setminus A$ has Euler-genus at most $g-1$ (this is at most
  $g-2$ if $A$ is two-sided).
\end{lemma}

This fits nicely in the bounded threateners program: Our algorithm
will iteratively take such a cycle $A$, create a random buffer $S$
around it, and recurse on the connected components of $G\setminus
S$. The base case is when the component is planar, then we may apply
our strong-diameter padding algorithm. Formally, in iteration $i$ take
a connected component $C_i$ in $G_i$, if $C_i$ is not planar, find a
cycle $A_i$ as in \lemmaref{lem:genus}. Let
$S_i=B_{G_i}(A_i,R_i\Delta)$ where
$R_i\propto\Texp_{[0,1/4]}(8\log g)$, set $G_{i+1}=G_i\setminus
S_i$. Each $S_i$ is partitioned to clusters by iteratively taking
cones centered at some of the points of $A_i$. If $C_i$ is planar,
invoke the decomposition scheme of \sectionref{sec:strong-diam}.

\begin{algorithm}\label{alg:genus}
\caption{\texttt{Genus-Partition}($G$,$\Delta$,$g$)}
\begin{algorithmic}[1]
\STATE Let $G_0 \leftarrow G$, $i=0$.
\STATE Let ${\cal C}\leftarrow \emptyset$.
\WHILE {$G_i$ is non-empty}
\STATE Let $C_i$ be a connected component of $G_i$.
\IF {$C_i$ is planar}
\STATE Let $P_i$ be a partition obtained by invoking \texttt{Strong-Random-Partition}($C_i,\Delta,5$). Add the clusters of $P_i$ to ${\cal C}$.
\STATE Set $G_{i+1}\leftarrow G_i\setminus \cup P_i$.
\ELSE
\STATE Let $A_i$ be cycle as in \lemmaref{lem:genus}.
\STATE Let $R_i\propto\Texp_{[0,1/4]}(8\log g)$.
\STATE Let $S_i=B_{G_i}(A_i,R_i\Delta)$.
\STATE \texttt{Create-Cones}($S_i,A_i,{\cal C}$). Add the resulting clusters to $\cC$.
\STATE Set $G_{i+1}\leftarrow G_i\setminus S_i$.
\ENDIF
\STATE $i\leftarrow i+1$.
\ENDWHILE
\RETURN $\cC$.
\end{algorithmic}
\end{algorithm}

We now turn to analyzing the algorithm. The fact that the resulting partition is strong-diameter $\Delta$-bounded follows from the fact that \texttt{Strong-Random-Partition} generates strong-diameter $\Delta$-bounded clusters, and by \lemmaref{lem:strong}, the cones are also strong-diameter $\Delta$-bounded (the proof of that lemma never used that $P$ is a shortest path, we only need that any point in $S_i$ is within distance $\Delta/4$ from $A_i$).

Fix some $z\in V$, $\gamma\le \delta$ for sufficiently small constant
$\delta$ (which is independent of $g$), and set $B_z = B_G(z,\gamma\Delta)$.
\begin{lemma}
  \label{lem:genus-cut-prob}
  The probability that the ball $B_z$ is cut by the above process is
  \[ \Pr[ B_z \mbox{ is cut }] \leq 1 - e^{-O(\gamma\log g)}~. \]

\end{lemma}
\begin{proof}
Let ${\cal E}_{\text{genus}}$ be the event that $B_z$ is first cut by some set $S_i$. Divide the event $\neg{\cal E}_{\text{genus}}$ into ${\cal F}_{\text{cone}}=\{\exists i,~B_z\subseteq S_i\}$ and ${\cal F}_{\text{planar}}=\{\exists i,~B_z\subseteq C_i\wedge C_i\text{ is planar}\}$. Let ${\cal E}_{\text{cone}}$ be the event that ${\cal F}_{\text{cone}}$ holds and also $B_z$ is first cut by a cone in the \texttt{Create-Cones}($S_i,A_i,\cC$), and finally let ${\cal E}_{\text{planar}}$ be the event that ${\cal F}_{\text{planar}}$ holds and also $B_z$ is cut while calling \texttt{Strong-Random-Partition} on a planar component containing $B_z$. We will bound each of the ${\cal E}$ events separately.

Assume w.l.o.g that non-planar components are chosen first, then the process until time $T$ (where all components are planar) is a cutting skeleton-process, with the graph sequence $G_1,\dots$, the skeletons $A_i$ and $K_i=S_i$, the parameters are $l=0$, $u=1/4$ and $b=2\log g$. Let $\threat_z=\{A_i  ~:~ i\in[T],~d_{G_i}(A_i,z)\le(1/4+\gamma)\Delta\}$. Note that by \lemmaref{lem:genus} there can be at most $g$ iterations (on components containing $z$) in which $z$ lies in a non-planar component, so $|\threat_z|\le g$. By \lemmaref{lem:cut}
\[
\Pr[{\cal E}_{\text{genus}}]\le (1-e^{-16\gamma\log g})\cdot(1+g/(e^{2\log g}-1))\le 1-e^{-32\gamma\log g}~,
\]
using that $\gamma\le 1/32$.
If $\Pr[\neg{\cal E}_{\text{genus}}]=p$, then $p\ge e^{-32\gamma\log g}$ and if $p_{\text{cone}}=\Pr[{\cal F}_{\text{cone}}]$ and $p_{\text{planar}}=\Pr[{\cal F}_{\text{planar}}]$ then since the events ${\cal F}_{\text{cone}}$, ${\cal F}_{\text{planar}}$ are disjoint, we have that
\begin{equation}\label{eq:p}
p=p_{\text{cone}}+p_{\text{planar}}~.
\end{equation}
By the first assertion of \theoremref{thm:stong-list}, there is a large constant $C$ such that
\[
\Pr[{\cal E}_{\text{planar}}]= p_{\text{planar}}\cdot O(\gamma)=p_{\text{planar}}(1-e^{-C\gamma})~,
\]
since $\gamma$ is sufficiently small.

Finally, we bound the probability of event ${\cal E}_{\text{cone}}$. Conditioning on $B_z\subseteq S_i$ for some $i$, we use a similar argument as in the proof of \lemmaref{lem:strong-cut-prob}, here we claim that there can be at most $18$ centers whose cone may intersect $B_z$. This is because if there are more, at least $10$ of them lie on one of the two  shortest path $A_i$ is comprised of, and using the argument appearing in the proof of \lemmaref{lem:strong-cut-prob}, it cannot be that all of them threaten $B_z$. Since $\alpha$ is chosen uniformly from an interval of length $\Delta/8$, the probability that any cone cuts $B_z$ is at most $2\gamma\Delta/(\Delta/8)$, thus by a union bound, using that $C$ is large enough,
\[
\Pr[{\cal E}_{\text{cone}}]= p_{\text{cone}}\cdot O(\gamma)=p_{\text{cone}}(1-e^{-C\gamma})~.
\]

Combining the three bounds, we obtain that the probability that $B_z$ is cut is at most

\begin{eqnarray*}
\Pr[{\cal E}_{\text{genus}}]+\Pr[{\cal E}_{\text{cone}}]+\Pr[{\cal E}_{\text{planar}}]&\le&1-p+p_{\text{cone}}(1-e^{-C\gamma})+p_{\text{planar}}(1-e^{-C\gamma})\\
&\stackrel{\eqref{eq:p}}{=}&1-p\cdot e^{-C\gamma}\\
&\le&1-e^{-32\gamma\log g}\cdot e^{-C\gamma}\\
&=&1-e^{-O(\gamma\log g)}~.
\end{eqnarray*}
\end{proof}

\section{Further Directions}

A clear open problem is to improve the $O(r)$-padded partition scheme for $K_r$-minor-free graphs to the optimal $O(\log r)$. A first step might be proving such a result for graphs of tree-width $r$ (recall that such graphs have a strong-diameter $O(\log r+\log\log n)$-padded partition).

\subsection*{Acknowledgments}

We are grateful to Alex Andoni and Daniel Berend for fruitful discussions. A.~Gupta and C. Gavoille thank Microsoft Research SVC for their kind hospitality. We also thank Arnold Filtser for pointing out an error in the 
original proof of Lemma~\ref{lem:threat}.

\bibliographystyle{alpha}
\bibliography{bib}

\newcommand{\etalchar}[1]{$^{#1}$}
\begin{thebibliography}{AGMW10}

\bibitem[ABN11]{ABN11}
Ittai Abraham, Yair Bartal, and Ofer Neiman.
\newblock Advances in metric embedding theory.
\newblock {\em Advances in Mathematics}, 228(6):3026 -- 3126, 2011.

\bibitem[AFH{\etalchar{+}}04]{AFHKTT}
Aaron Archer, Jittat Fakcharoenphol, Chris Harrelson, Robert Krauthgamer, Kunal
  Talwar, and {\'E}va Tardos.
\newblock Approximate classification via earthmover metrics.
\newblock In {\em Proceedings of the 15th ACM-SIAM Symposium on Discrete
  Algorithms (SODA)}, pages 1079--1087, New York, 2004. ACM.

\bibitem[AGG{\etalchar{+}}14]{AGGNT14}
Ittai Abraham, Cyril Gavoille, Anupam Gupta, Ofer Neiman, and Kunal Talwar.
\newblock Cops, robbers, and threatening skeletons: Padded decomposition for
  minor-free graphs.
\newblock In {\em Proceedings of the Forty-sixth Annual ACM Symposium on Theory
  of Computing}, STOC '14, pages 79--88, New York, NY, USA, 2014. ACM.

\bibitem[AGMW10]{AGMW10}
Ittai Abraham, Cyril Gavoille, Dahlia Malkhi, and Udi Wieder.
\newblock Strong-diameter decompositions of minor free graphs.
\newblock {\em Theory Comput. Syst.}, 47(4):837--855, 2010.

\bibitem[And86]{Andreae86}
Thomas Andreae.
\newblock On a pursuit game played on graphs for which a minor is excluded.
\newblock {\em J. Combin. Theory Ser. B}, 41(1):37--47, 1986.

\bibitem[AST90]{AST90}
Noga Alon, Paul Seymour, and Robin Thomas.
\newblock A separator theorem for nonplanar graphs.
\newblock {\em J. Amer. Math. Soc.}, 3(4):801--808, 1990.

\bibitem[Awe85]{A85}
Baruch Awerbuch.
\newblock Complexity of network synchronization.
\newblock {\em J. ACM}, 32(4):804--823, October 1985.

\bibitem[Bar96]{Bar96}
Y.~Bartal.
\newblock Probabilistic approximation of metric spaces and its algorithmic
  applications.
\newblock In {\em Proceedings of the 37th Annual Symposium on Foundations of
  Computer Science}, FOCS '96, pages 184--, Washington, DC, USA, 1996. IEEE
  Computer Society.

\bibitem[BK96]{BK96}
Hans~Leo Bodlaender and Ton Kloks.
\newblock Efficient and constructive algorithms for the pathwith and treewidth
  of graphs.
\newblock {\em Journal of Algorithms}, 21(2):358--402, 1996.

\bibitem[BLR10]{BiswalLR10}
Punyashloka Biswal, James~R. Lee, and Satish Rao.
\newblock Eigenvalue bounds, spectral partitioning, and metrical deformations
  via flows.
\newblock {\em J. ACM}, 57(3), 2010.

\bibitem[BLS10]{BorradaileLS10}
Glencora Borradaile, James~R. Lee, and Anastasios Sidiropoulos.
\newblock Randomly removing {$g$} handles at once.
\newblock {\em Comput. Geom.}, 43(8):655--662, 2010.

\bibitem[BLT07]{BLT07}
Costas Busch, Ryan LaFortune, and Srikanta Tirthapura.
\newblock Improved sparse covers for graphs excluding a fixed minor.
\newblock In {\em Proceedings of the twenty-sixth annual ACM symposium on
  Principles of distributed computing}, PODC '07, pages 61--70, New York, NY,
  USA, 2007. ACM.

\bibitem[CdV10]{Colin10}
\'Eric Colin~de Verdi\'ere.
\newblock Shortest cut graph of a surface with prescribed vertex set.
\newblock In {\em $18^{th}$ Annual European Symposium on Algorithms (ESA)},
  volume 6347 of Lecture Notes in Computer Science, pages 100--111. Springer,
  September 2010.

\bibitem[CKR05]{CKR01-zero}
Gruia Calinescu, Howard Karloff, and Yuval Rabani.
\newblock Approximation algorithms for the 0-extension problem.
\newblock {\em SIAM J. Comput.}, 34(2):358--372, 2004/05.

\bibitem[Die00]{diestel}
Reinhard Diestel.
\newblock {\em Graph theory}, volume 173 of {\em Graduate Texts in
  Mathematics}.
\newblock Springer-Verlag, New York, second edition, 2000.

\bibitem[FHL08]{FeigeHL08}
Uriel Feige, MohammadTaghi Hajiaghayi, and James~R. Lee.
\newblock Improved approximation algorithms for minimum weight vertex
  separators.
\newblock {\em SIAM J. Comput.}, 38(2):629--657, 2008.

\bibitem[FRT04]{FRT03}
Jittat Fakcharoenphol, Satish Rao, and Kunal Talwar.
\newblock A tight bound on approximating arbitrary metrics by tree metrics.
\newblock {\em J. Comput. System Sci.}, 69(3):485--497, 2004.

\bibitem[FT03]{FT03}
Jittat Fakcharoenphol and Kunal Talwar.
\newblock An improved decomposition theorem for graphs excluding a fixed minor.
\newblock {\em RANDOM-APPROX}, pages 36--46, 2003.

\bibitem[GKL03]{GKL03}
Anupam Gupta, Robert Krauthgamer, and James~R.\ Lee.
\newblock Bounded geometries, fractals, and low--distortion embeddings.
\newblock In {\em FOCS}, pages 534--543, 2003.

\bibitem[GS01]{GS01}
Geoffrey~R. Grimmett and David~R. Stirzaker.
\newblock {\em Probability and random processes}.
\newblock Oxford University Press, New York, third edition, 2001.

\bibitem[GTW13]{GTW13}
Anupam Gupta, Kunal Talwar, and David Witmer.
\newblock Sparsest cut on bounded treewidth graphs: Algorithms and hardness
  results.
\newblock In {\em Proceedings of the Forty-fifth Annual ACM Symposium on Theory
  of Computing}, STOC '13, pages 281--290, New York, NY, USA, 2013. ACM.

\bibitem[IS07]{IndykS07}
Piotr Indyk and Anastasios Sidiropoulos.
\newblock Probabilistic embeddings of bounded genus graphs into planar graphs.
\newblock In {\em Symposium on Computational Geometry}, pages 204--209, 2007.

\bibitem[KLPT09]{KelnerLPT09}
Jonathan~A. Kelner, James~R. Lee, Gregory~N. Price, and Shang-Hua Teng.
\newblock Higher eigenvalues of graphs.
\newblock In {\em FOCS}, pages 735--744, 2009.

\bibitem[KMR08]{KMR08}
Ken-ichi Kawarabayashi, Bojan Mohar, and Bruce~A. Reed.
\newblock A simpler linear time algorithm for embedding graphs into an
  arbitrary surface and the genus of graphs of bounded tree-width.
\newblock In {\em $49^{th}$ Annual IEEE Symposium on Foundations of Computer
  Science (FOCS)}, pages 771--780. IEEE Computer Society Press, October 2008.

\bibitem[KPR93]{KPR93}
Philip~N. Klein, Serge~A. Plotkin, and Satish Rao.
\newblock Excluded minors, network decomposition, and multicommodity flow.
\newblock In {\em STOC}, pages 682--690, 1993.

\bibitem[Lee13]{JRL-blog}
James~R.\ Lee.
\newblock Open question recap, February 2013.
\newblock http://tcsmath.wordpress.com/2013/02/25/open-question-recap/.

\bibitem[LGT12]{LeeGT12}
James~R. Lee, Shayan~Oveis Gharan, and Luca Trevisan.
\newblock Multi-way spectral partitioning and higher-order cheeger
  inequalities.
\newblock In {\em STOC}, pages 1117--1130, 2012.

\bibitem[LN05]{LN03}
James~R. Lee and Assaf Naor.
\newblock Extending {L}ipschitz functions via random metric partitions.
\newblock {\em Invent. Math.}, 160(1):59--95, 2005.

\bibitem[LS93]{LS93}
Nathan Linial and Michael Saks.
\newblock Low diameter graph decompositions.
\newblock {\em Combinatorica}, 13(4):441--454, 1993.
\newblock (Preliminary version in {\em 2nd SODA}, 1991).

\bibitem[LS10]{LeeS10}
James~R. Lee and Anastasios Sidiropoulos.
\newblock Genus and the geometry of the cut graph.
\newblock In {\em SODA}, pages 193--201, 2010.

\bibitem[Mat02]{mat-book}
Ji{\v{r}}\'{\i} Matou{\v{s}}ek.
\newblock {\em Lectures on discrete geometry}, volume 212 of {\em Graduate
  Texts in Mathematics}.
\newblock Springer-Verlag, New York, 2002.

\bibitem[Moh99]{Mohar99}
Bojan Mohar.
\newblock A linear time algorithm for embedding graphs in an arbitrary surface.
\newblock {\em SIAM Journal on Discrete Mathematics}, 12(1):6--26, 1999.

\bibitem[MT01]{MT01}
Bojan Mohar and Carsten Thomassen.
\newblock {\em Graphs on Surfaces}.
\newblock The Johns Hopkins university Press, 2001.

\bibitem[PRS94]{PRS94}
Serge Plotkin, Satish Rao, and Warren~D. Smith.
\newblock Shallow excluded minors and improved graph decompositions.
\newblock In {\em Proceedings of the {F}ifth {A}nnual {ACM}-{SIAM} {S}ymposium
  on {D}iscrete {A}lgorithms ({A}rlington, {VA}, 1994)}, pages 462--470, New
  York, 1994. ACM.

\bibitem[Rab03]{Yuri03}
Yuri Rabinovich.
\newblock On average distortion of embedding metrics into the line and into
  $\ell_1$.
\newblock In {\em Proceedings of the thirty-fifth ACM symposium on Theory of
  computing}, pages 456--462. ACM Press, 2003.

\bibitem[Rao99]{Rao99}
Satish~B. Rao.
\newblock Small distortion and volume preserving embeddings for planar and
  {Euclidean} metrics.
\newblock In {\em SOCG}, pages 300--306, 1999.

\bibitem[Ree92]{Reed92}
Bruce~A. Reed.
\newblock Finding approximate separators and computing treewidth quickly.
\newblock In {\em $24^{th}$ Annual ACM Symposium on Theory of Computing
  (STOC)}, pages 221--228. ACM Press, 1992.

\bibitem[RS03]{RS03}
Neil Robertson and Paul~D. Seymour.
\newblock Graph minors. {XVI}. {E}xcluding a non-planar graph.
\newblock {\em Journal of Combinatorial Theory, Series B}, 89(1):43 -- 76,
  2003.

\bibitem[Sid10]{Sidiropoulos10}
Anastasios Sidiropoulos.
\newblock Optimal stochastic planarization.
\newblock In {\em FOCS}, pages 163--170, 2010.

\bibitem[WN11]{W11}
Christian Wulff-Nilsen.
\newblock Separator theorems for minor-free and shallow minor-free graphs with
  applications.
\newblock In {\em FOCS}, pages 37--46, 2011.

\end{thebibliography}

\end{document}